\documentclass[runningheads]{llncs}
\usepackage{graphicx}
\usepackage{graphics} % for EPS, load graphicx instead
\usepackage{url}  %Required
\usepackage[font=small,skip=0pt]{caption}
\usepackage[font=small,skip=0pt]{subcaption}
\usepackage{amsmath,amssymb}
\usepackage{mathrsfs}

\begin{document}

\title{\normalsize{ \bf SUMMARIZED: Efficient Framework for Analyzing Multidimensional Process Traces under Edit-distance Constraint}}
\titlerunning{ }
\author{\small Phuong Nguyen\inst{1} \and
Vatche Ishakian\inst{2} \and
Vinod Muthusamy\inst{2} \and
Aleksander Slominski\inst{2}}
\authorrunning{P. Nguyen et al.}
\institute{University of Illinois Urbana Champaign \email{pvnguye2@illinois.edu}
\and
IBM Research
\email{vatchei@ibm.com,\{vmuthus,aslom\}@us.ibm.com}}
\maketitle

\begin{abstract}

Domains such as scientific workflows and business processes exhibit data models with complex relationships between objects. This relationship is typically represented as sequences, where each data item is annotated with multi-dimensional attributes. There is a need to analyze this data for operational insights. For example, in business processes, users are interested in clustering process traces into smaller subsets to discover less complex process models. This requires expensive computation of similarity metrics between sequence-based data. Related work on dimension reduction and embedding methods do not take into account the multi-dimensional attributes of data, and do not address the interpretability of data in the embedding space (i.e., by favoring vector-based representation). In this work, we introduce \textsc{Summarized}, a framework for efficient analysis on sequence-based multi-dimensional data using intuitive and user-controlled summarizations. We introduce summarization schemes that provide tunable trade-offs between the quality and efficiency of analysis tasks and derive an error model for summary-based similarity under an edit-distance constraint. %a popular similarity measure for sequences, Evaluations using real-world datasets show that \textsc{Summarized} is effective in applications such as similarity search and process traces clustering.
Evaluations using real-world datasets show the effectives of our framework.

%\keywords{First keyword  \and Second keyword \and Another keyword.}
\end{abstract}

\section{Introduction}
\label{sec:intro}
Application domains, such as business processes and scientific
workflows, exhibit data models in the form of multi-dimensional
sequence of objects. For example, in business processes, given an
underlying business process model represented as a directed acyclic
graph of activities, the traces generated from the execution of the
model are regarded as instances of the underlying model. Each trace
consists of a sequence of activities sorted by time, where each
activity in the trace appears in the process model and may be repeated\footnote{\scriptsize In this paper, we use
\textit{trace}, \textit{process trace}, and \textit{sequence}
interchangeably to refer to an instance of a process.}.
Figure~\ref{fig:motivating_example1} shows an example of a loan
application process model. The highlighted activities in the figure
represent a possible execution trace of the model.  In addition to the
sequential structure, each activity also contains multi-dimensional
attributes.  For example, an activity in the loan application process
can contain information about the person  who performs the activity, the group to which she belongs, and the department
responsible for the activity. % For example, an activity in the loan application process  can contain information about the person and the department that are  responsible for the activity, the person who performs the activity,  and the group to which she belongs.
In another domain, provenance data captured from the execution of
scientific workflows are also in the form of multi-dimensional
sequences.  Figure~\ref{fig:motivating_example2} shows a sample trace
of a semiconductor manufacturing process, where each activity can
consist of additional information, such as the sector where the
activity is performed and the person responsible for that
activity.

\begin{figure}[tp]
	\begin{minipage}{0.56\textwidth}
	\centering
	\includegraphics[width=1\linewidth]{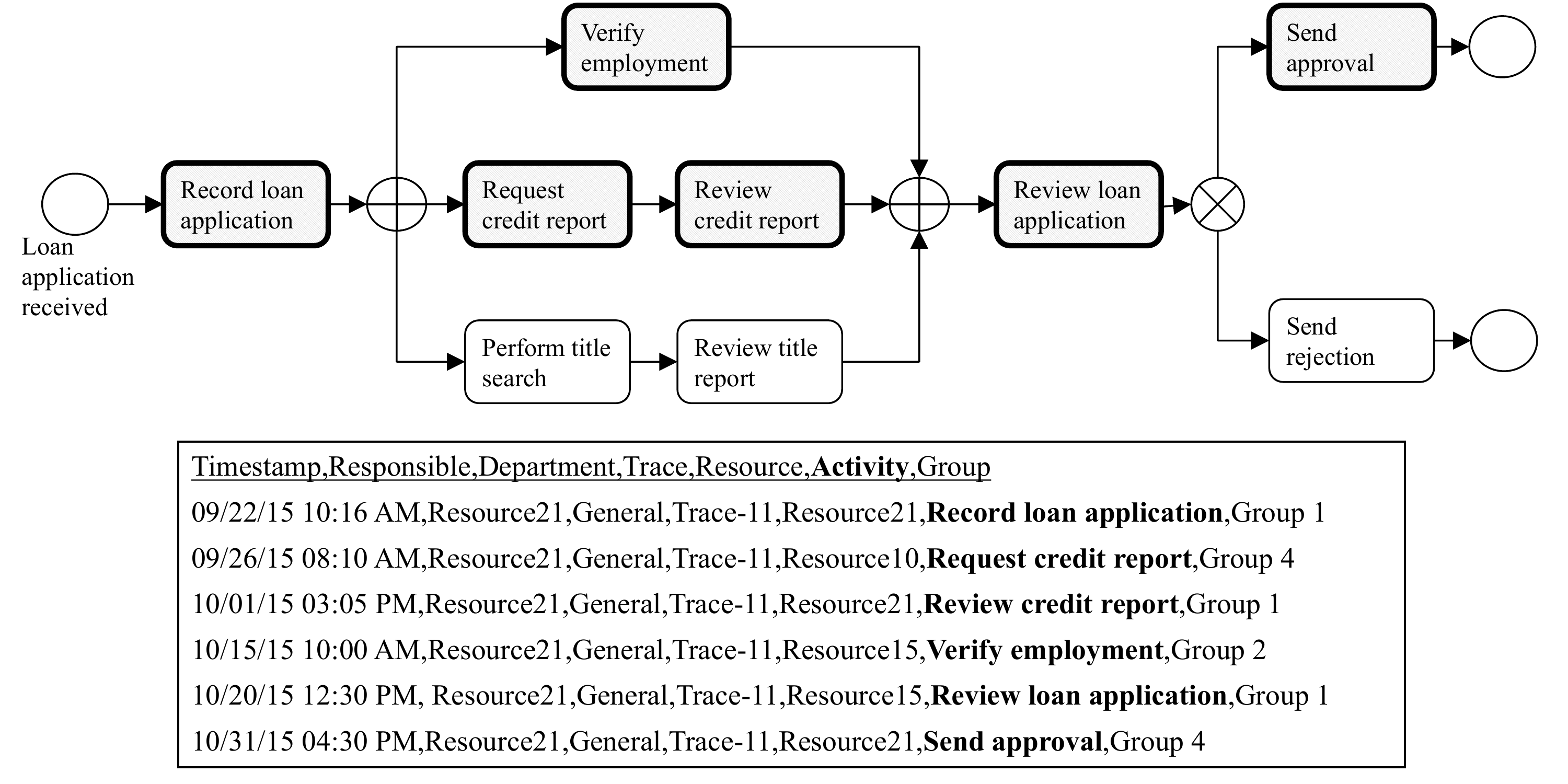}
  \caption{Example 1: loan application process and a sample trace.}
  \label{fig:motivating_example1}
  \end{minipage}
  ~
	\begin{minipage}{0.43\textwidth}
		\centering
		\includegraphics[width=1\linewidth]{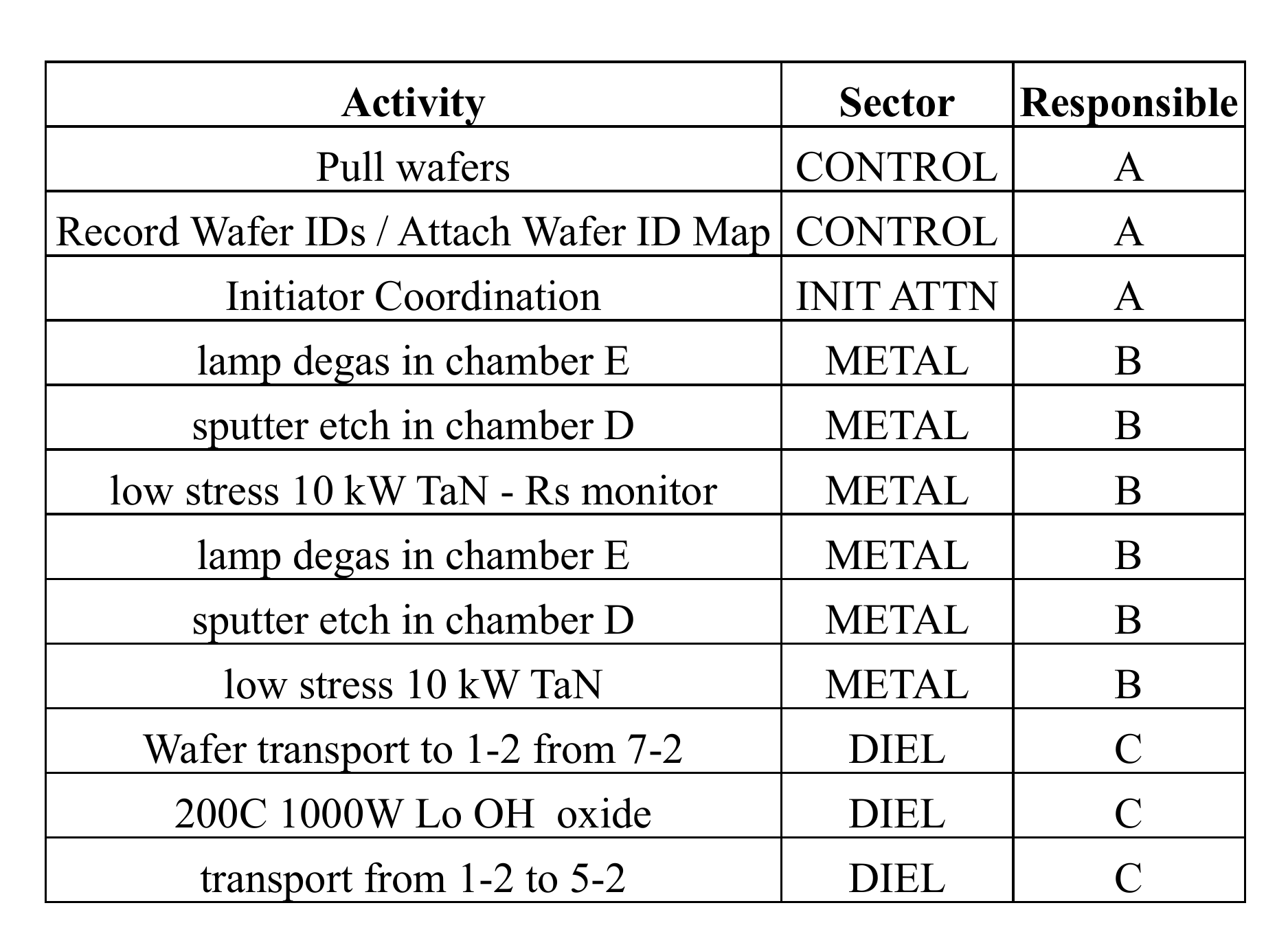}
	  \caption{Example 2: sample trace of a scientific workflow.}
	  \label{fig:motivating_example2}
	\end{minipage}
\vspace{-0.3in}
\end{figure}

With the popularity of such applications, there are increasing needs to analyze
the data for operational insights, and there are many efforts to apply machine
learning techniques in the business process management
field~\cite{vanderAalst2011PMD,masellis2017verification}. As business
models mined from complete process traces are often complex and difficult to
comprehend~\cite{de2007process}, users are interested in clustering process
traces into smaller subsets and applying process discovery
algorithms~\cite{van2004workflow} on each subset. The models discovered using
only the traces in a cluster tend to be both less complex and more accurate
since there is less diversity among the traces within a cluster. In another
example, scientists are interested in querying the provenance data of scientific
workflow executions to look for previous executions of a workflow that are
similar to the one in the query, again using models trained on historical
executions.

% In this paper, we study multi-dimensional sequence data obtained from the execution of business processes, or scientific workflow's provenance data.
Analyzing multi-dimensional sequence data poses a number of challenges. The
\textit{first} challenge is in terms of computational complexity of data
analysis. For example, edit-distance is often used to capture the similarity
between sequences~\cite{bose2009context}. Since edit-distance is quadratic to
the sequence length and each sequence can consist of hundreds of data items
(e.g., in business processes), it is computationally expensive to compute the
similarities between sequences. This is especially challenging when dealing with
large datasets and in applications such as traces clustering, where a lot of
similarity computations need to be calculated. This complexity can also cause
long application delays that affect interactive applications, such as similarity
search, where users interact directly with the application and expect to get the
results in a timely manner. The \textit{second} challenge is to combine
multi-dimensional attributes of data with the sequential structure between data
objects into a unified approach.  Edit-distance, for example, only concerns with
counting the minimum number of basic operations required to transform one
sequence of activities into the other, without considering the attributes of
activities.

%There have been a number of efforts proposed to address the above  challenges (see Section~\ref{sec:relatedwork} for details). For  example, \textit{embedding methods}  \cite{hristescu1999cluster}~\cite{faloutsos1995fastmap}\cite{wang2000index} represent sequences in a vector space and use a more efficient  measure, such as Cosine similarity, to capture similarity between sequences.  Related work in \textit{subsequence mapping and sequence retrieval} often use a set of reference  sequences~\cite{papapetrou2009reference}\cite{venkateswaran2006reference}  to perform more efficient search over complex sequences. These  methods, however, avoid the sequential relationship between data items  from the original sequences, and do not consider the multi-dimensional  attributes of each data item.
In this paper, we 
introduce \textsc{Summarized}, a framework for the efficient analysis
of multi-dimensional sequence data under edit-distance constraint.
We focus on analysis tasks that are based on
edit-distance similarity because it is a widely used measure for sequences.
Instead of performing computationally expensive analysis on the
original high-dimensional data, we transform the data into a summary
space that has fewer dimensions, so that more efficient analysis can
be applied. To incorporate multi-dimensional attributes of data items
into the analysis, we introduce summarization schemes that allow users
to select attributes as the summarization criteria. In addition to
attribute-based summarizations, which produce summaries of good
semantics but are limited in giving users control over the resolution
of summaries (and thus, the efficiency), we also introduce topic-based
summarization that enables the flexible trade-off between quality and
efficiency of analysis tasks on summaries. In addition, we develop an
error model for the edit-distance measure in the summary space to
provide theoretical guarantees for the results of analysis tasks on
summaries.

%Our contributions in this paper are as follows:

%\begin{itemize}
%\item We present a novel framework for analyzing multi-dimensional process traces using intuitive and user-controlled summarizations.
%\item We introduce various summarization schemes that offer flexible trade-off between the quality and efficiency of analysis tasks.
%\item We derive an error model for summary-based similarity under an edit-distance constraint.
%\end{itemize}

\section{Related Work}
\label{sec:relatedwork}

There have been active research on \textit{subsequence mapping and sequence retrieval}, especially with biological sequences data \cite{roy2012massive}. To support efficient mapping and retrieval, one of the common approaches is summarize original sequences using q-grams \cite{burkhardt1999q}\cite{li2007vgram}\cite{kim2016hobbes3} and measure the similarity between two sets of q-grams. Another common approach is reference-based method. For example, \cite{papapetrou2009reference}\cite{venkateswaran2006reference} filter the results to a query sequence using precomputed distance between sequences in the database and a set reference sequences. DRESS \cite{kotsifakos2015dress} uses the most frequent codewords as references to identify a set candidate matches of a query. These work, however, do not consider sequences of multi-dimensional attributes of data items in the original sequences. In addition, both q-gram-based and reference-based methods do not preserve the sequential relationship between data items in the original sequences, and thus, do not support similarity measure under edit-distance constraint.

Another area  of related work is on \textit{graph similarity and mining}, where sequence is a special case. Since graph edit-distance is also very computationally expensive, most of the work try to transform the original graph to a more compact representation before measuring similarity. One common transformation approach is based on graph' substructures, such as stars \cite{zeng2009comparing}, trees \cite{zheng2013graph}, branches \cite{li2017efficient}, paths \cite{zhao2012efficient}, or k-shingle \cite{manzoor2016fast}. Recently, there has been effort on solving graph edit-distance using binary linear programming \cite{lerouge2017new}. While some of these work provide error bound on graph edit-distance on substructure space, they only consider homogeneous graphs. In addition, a major issue with this group of related work is that graphs lose their interpretability and graphical representation after being transformed into substructure representation (i.e., graphs are either represented as bags-of-substructures \cite{zeng2009comparing}\cite{zheng2013graph}, or numeric vector \cite{manzoor2016fast}).

In terms of the \textit{intuitive and interpretable summarization of graphs} \cite{liu2018graph}\cite{khan2017summarizing}, Zhao et al. \cite{zhao2011graph} introduce Graph Cube model that supports OLAP queries effectively on large multidimensional networks. Such a model can be used to produce interpretable summaries of original graph, in form of aggregate graphs, by performing cuboid queries. Tian et al. \cite{tian2008efficient} introduce two database-style operations to summarize multi-dimensional graphs: one produces a summary by grouping nodes based on attributes and relationships, while the other allows users to control the resolutions of summaries. Chen et al. \cite{chen2009mining} show that random summaries can help effectively reduce the size of original graph and at the same time, are capable of mining frequent graph patterns. In this paper, besides using explicit attributes, we also leverage the implicit topics as summarization criteria. We also show that, different from general graphs, random summarization on sequences, although produces good effectiveness, suffers from efficiency.

\textit{Embedding methods} \cite{hristescu1999cluster}\cite{faloutsos1995fastmap}\cite{wang2000index} have been used to improve efficiency of similarity search on complex data (see \cite{hjaltason2003properties} for an excellent survey).  However, there are only a few embedding approaches that could guarantee important property of similarity measure on the embedding space, such as contractive property. For example, it might require the similarity measure between data on the embedding space is from a specific family of measure (e.g., Minkowski metric \cite{hristescu1999cluster}). Another major drawback of embedding techniques is that they transform original sequences into vector-based representation, and thus, do not maintain the sequential relationship between data items on the new representation. In this paper, we formally define summarization schemes on sequences and show that contractive property of similarity measure on summary space under edit-distance constraint can be guaranteed.

There have been efforts to address \textit{the scalability issue} in process mining and business process analysis \cite{sayal2002business}. However, these efforts mainly focus on process model discovery of large, complex traces \cite{leemans2015scalable}\cite{van2004workflow}\cite{weijters2006process}. There are also related work on using vector space-based dimensional reduction to improve the performance of traces clustering \cite{song2013comparative}\cite{nguyen2016process}. In this paper, we focus on improving efficiency of traces similarity search and traces clustering \textit{under edit-distance constraint}. There is also related work to perform process discovery on large-scale dataset by using Map-Reduce \cite{evermann2016scalable}. Our work can be used in combination with the related efforts. For example, once traces are clustered into smaller subsets, efficient process discovery algorithms \cite{grigori2001improving} can be applied to each subset.

%There have been recent attempts to apply \textit{deep learning} to graph structures, including to predict process traces~\cite{EvermannRehseFettke:2017}, and to cluster graphs~\cite{tian2014learning}. Unlike this paper, the embeddings discovered by the neural networks in these approaches are not interpretable.

\section{Framework}
\label{sec:framework}

Figure~\ref{fig:overview} highlights the motivation for designing the \textsc{Summarized} framework. We assume the existence of an original dataset,which consists of a set of process traces or logs of scientific workflow executions. Running an analysis, which would typically be computationally expensive due to the high-dimentionality of the data, provides results which are deemed as exact or ``ground truth" answer.

\begin{figure}[tp]
	\centering
	\includegraphics[width=0.6\linewidth]{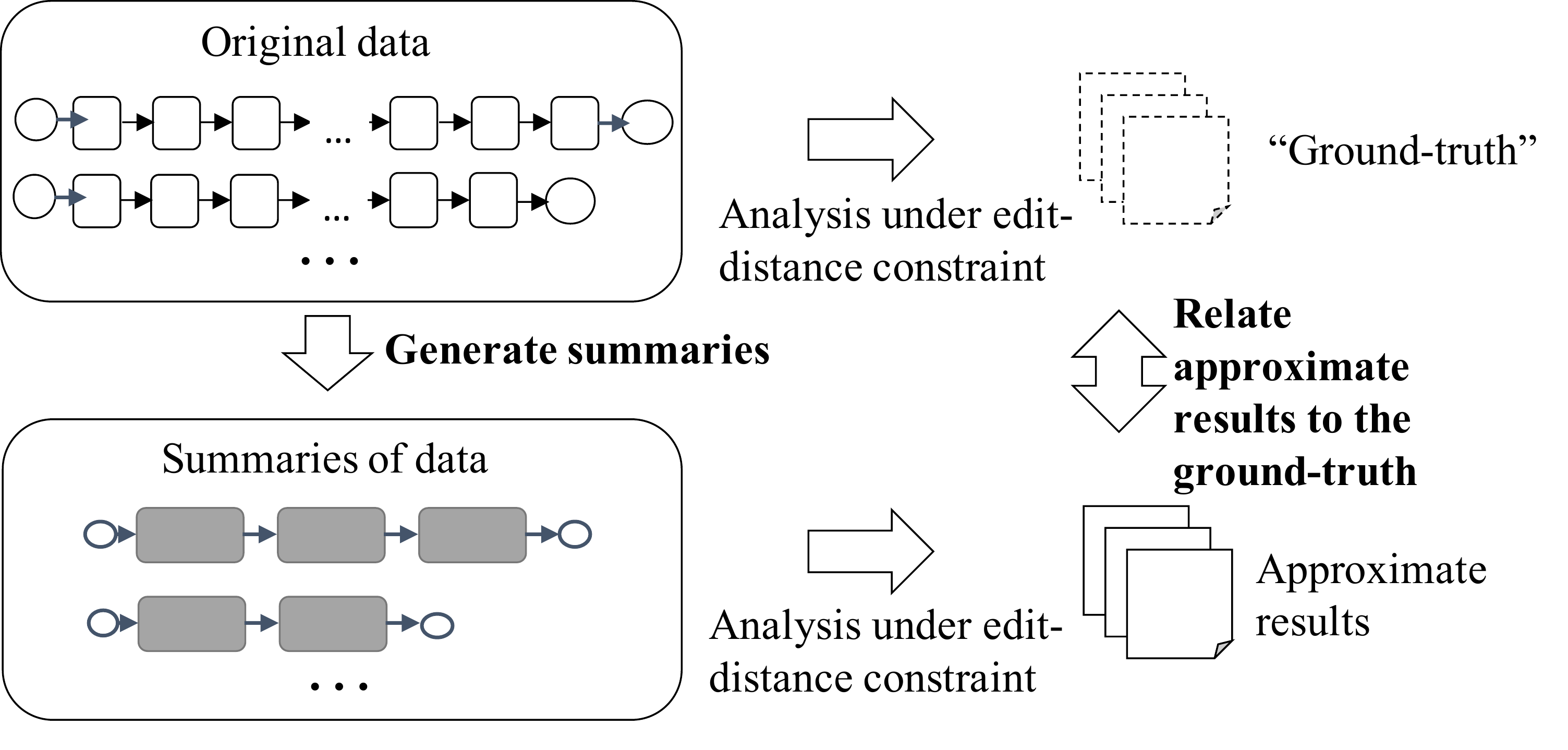}
  \caption{Overview of \textsc{Summarized}'s approach.}
  \label{fig:overview}
  \vspace{-0.25 in}
\end{figure}
The key principal of our framework is to transform the original data into a new space with fewer number of dimensions, thus avoiding the computationally expensive analysis on original dataset. The resulting output, is inherently different than the ``ground truth", is known as an approximate result. To demonstrate the practicality of our framework, we need to address the following two challenges: (1) How to generate summaries of data in a controlled and intuitive manner, and (2) How to relate the approximate results on summaries to the results on original data?

%An overview of our \textsc{Summarized} framework is presented in Figure~\ref{fig:overview}. The key principle of \textsc{Summarized} is to transform the data into summaries of fewer number of dimensions to avoid computationally expensive analysis on original high-dimensional data. With \textsc{Summarized}, we seek to answer the following fundamental questions:

% \begin{itemize}
% \item How to generate summaries of data in a controlled and intuitive manner?
% \item How to relate the approximate results on summaries to the results on original data?
% \end{itemize}

For the first challenge, many sequence- and graph-based (in which sequence is a special form) summarization methods generate summaries using statistics, patterns, or sub-structures of the data. Thus, the resulted summaries are often difficult to interpret by users as they lack the structural semantic connection with the original representation. The lack of structural semantics of summaries also prevents analysis tasks that rely on the structural information (such as edit distance-based analysis, whose results are easy-to-interpret by users) to be performed on summary space. Finally, under currently existing methods, users do not have {\it much} control over the summaries will be generated. As a result, it is difficult to integrate user expertise and feedback into the summarization process to guide the data analysis.
%In Section~\ref{sec:summarization}, we formally define summarization on sequences and, in Section~\ref{sec:schemes}, we introduce some summarization schemes that are intuitive and give users more control over the resulted summaries.

For the second challenge, since the analysis results on data summaries might not be the same to those on original data, it is important to understand the relationship between the two results and for all practical purposes, provide guarantees about the quality of results obtained from summaries.

%In Section~\ref{sec:similarity}, we present an error model for summary-based similarity measure under edit-distance constraint and show that it provides quality guarantee over the results of similarity search task. We also show via empirical study in Section~\ref{sec:evaluation} that summary-based clustering of sequence data produces results with comparable accuracy, compared with when using original data, while having better efficiency.

To address the above mentioned challenges, in the remaining sections, we define sequential-order-preserving summarization on sequences and introduce several summarization schemes that are intuitive and give users more control over the resulted summaries. We also formally present an error model for summary-based similarity measure under edit-distance constraint and show that it provides quality guarantee over the results of similarity search task. %Using empirical studies on real trace dataset, we show that our summary-based clustering of sequence data produces results with comparable accuracy, compared with when using original data, while having better efficiency.

\section{Definitions}
\label{sec:definitions}

% This section provides basic definitions of the notion of multidimensional sequence and summarization of sequences. We start by defining a multidimensional set of objects and \textit{multidimensional sequences}.
% \begin{definition} (\textit{\textbf{Multidimensional Set}})
% A multidimensional set $\mathscr{O}$ is defined as a set of objects $\mathbb{O}$ and a set of associated attributes $\mathbb{A} = (\mathcal{A}_1, \mathcal{A}_2, ..., \mathcal{A}_{|\mathbb{A}|})$: $\mathscr{O} = \langle\mathbb{O}, \mathbb{A}\rangle$, each object $o \in \mathbb{O}$ is defined as a tuple: $o = (\mathcal{A}_1(o), \mathcal{A}_2(o), ..., \mathcal{A}_{|\mathbb{A}|}(o))$, in which each $i$-th dimension corresponds to the value of attribute $\mathcal{A}_i$ of $o$, denoted as $\mathcal{A}_i(o)$.
% \label{def:sets}
% \end{definition}

%\begin{definition} (\textit{\textbf{Sequence}})
%A sequence $\mathbf{p}$ of size $n$ over a set of objects in $\mathscr{O}$ is defined as an ordered set of $n$ objects $\mathbf{p} = (p_1, p_2, ..., p_m), p_i \in \mathscr{O}, 1 \leq i \leq m$.
%\label{def:sequences}
%\end{definition}

% , where $p_{i} \neq p_{i+1} \forall 1 \leq i \leq m-1$

%Given the above definition of multidimensional set, we now can define the notion of \textit{multidimensional sequences}:

This section provides basic definitions of the notion of multidimensional sequence and summarization of sequences. We define a multidimensional set $\mathscr{O}$ as a set of objects $\mathbb{O}$ and a set of associated attributes $\mathbb{A} = (\mathcal{A}_1, \mathcal{A}_2, ..., \mathcal{A}_{|\mathbb{A}|})$: $\mathscr{O} = \langle\mathbb{O}, \mathbb{A}\rangle$, each object $o \in \mathbb{O}$ is defined as a tuple: $o = (\mathcal{A}_1(o), \mathcal{A}_2(o), ..., \mathcal{A}_{|\mathbb{A}|}(o))$, in which each $i$-th dimension corresponds to the value of attribute $\mathcal{A}_i$ of $o$, denoted as $\mathcal{A}_i(o)$.

% \begin{definition} (\textit{\textbf{Multidimensional Set}})
%  A multidimensional set $\mathscr{O}$ is defined as a set of objects $\mathbb{O}$ and a set of associated attributes $\mathbb{A} = (\mathcal{A}_1, \mathcal{A}_2, ..., \mathcal{A}_{|\mathbb{A}|})$: $\mathscr{O} = \langle\mathbb{O}, \mathbb{A}\rangle$, each object $o \in \mathbb{O}$ is defined as a tuple: $o = (\mathcal{A}_1(o), \mathcal{A}_2(o), ..., \mathcal{A}_{|\mathbb{A}|}(o))$, in which each $i$-th dimension corresponds to the value of attribute $\mathcal{A}_i$ of $o$, denoted as $\mathcal{A}_i(o)$.
% \label{def:sets}
% \end{definition}

% \begin{definition} (\textit{\textbf{Multidimensional Sequence}})
% A sequence $\mathbf{p}$ of size $m$ on a multidimensional set $\mathscr{O}$ is defined as an ordered set of $m$ objects in $\mathscr{O}$: $\mathbf{p} = (p_1, p_2, ..., p_m), p_i \in \mathscr{O}, 1 \leq i \leq m$.
% \label{def:multi-sequences}
% \end{definition}

A \textit{Multidimensional Sequence} $\mathbf{p}$ of size $m$ on a multidimensional set $\mathscr{O}$ is defined as an ordered set of $m$ objects in $\mathscr{O}$: $\mathbf{p} = (p_1, p_2, ..., p_m), p_i \in \mathscr{O}, 1 \leq i \leq m$.
We denote $\iota_{\mathbf{p}}(p)$ as the \textit{index}, or position, of an object $p$ in a sequence $\mathbf{p}$. In the above definition, $\iota_{\mathbf{p}}(p_i) = i, \forall 1 \leq i \leq m$. For example, Figure~\ref{fig:motivating_example2} presents a sequence of objects defined on a multidimensional set with three attributes: \textit{Activity}, \textit{Sector}, and \textit{Responsible}.

Our interest is in different forms of summarization of multidimensional sequences to improve efficiency of sequence analysis. Before defining summarization of sequences, we define the notion of many-to-one mapping of objects between multidimensional sets as an object mapping function $\mathit{f}$ from an original multidimensional set $\mathscr{O}$ to a summary set $\mathscr{S}$, $\mathit{f}: \mathscr{O} \rightarrow \mathscr{S}$, so that for each $p \in \mathscr{O}, \exists! s \in \mathscr{S}: s = \mathit{f}(p)$.

% \begin{definition} %(\textit{\textbf{Many-to-one Mapping}})
% A many-to-one mapping is defined as an object mapping function $\mathit{f}$ from an original multidimensional set $\mathscr{O}$ to a summary set $\mathscr{S}$, $\mathit{f}: \mathscr{O} \rightarrow \mathscr{S}$, so that for each $p \in \mathscr{O}, \exists! s \in \mathscr{S}: s = \mathit{f}(p)$.
% \label{def:manytoone}
% \end{definition}

Next, we define summarization of sequences based on many-to-one mapping $\mathit{f}$, called $\mathit{f}$-\textit{summarization}:

\begin{definition} %(\textbf{\textit{$\mathit{f}$-Summarization}})
A $\mathit{f}$-summarization of a sequence $\mathbf{p}$ on $\mathscr{O}$ is defined as a summary sequence $\mathbf{s}$ on $\mathscr{S}$, denoted as $\mathbf{s} = \mathit{f}(\mathbf{p})$, where each object $p \in \mathbf{p}$ is replaced by its many-to-one mapping $\mathit{f}$: $s = \mathit{f}(p)$, while retaining the same index $\iota_{\mathbf{s}}(s) := \iota_{\mathbf{p}}(p)$.
\label{def:fsummary}
\end{definition}

A summarization of a sequence is said to preserve the sequential relationship from the original sequence if it satisfies the following definition:

\begin{definition} %(\textbf{\textit{Sequential Preserving Summarization}})
A $\mathit{f}$-summarization of a sequence $\mathbf{p}$, denoted as $\mathbf{s} = \mathit{f}(\mathbf{p})$, is a sequential preserving summarization of $\mathbf{p}$ if: $\forall p, p' \in \mathbf{p}$, if $\iota_{\mathbf{p}}(p) < \iota_{\mathbf{p}}(p')$, then $\iota_{\mathbf{s}}(s) \leq \iota_{\mathbf{s}}(s')$, with $s = \mathit{f}(p), s' = \mathit{f}(p')$.
\label{def:seqpreservation}
\end{definition}

By retaining the indices of objects in the original sequence, $\mathit{f}$-summarization ({\it c.f, Definition~\ref{def:fsummary}}) preserves sequential relationships, which is vital in improving the efficiency of sequence analysis. Therefore, we define the notion of \textit{reduced $\mathit{f}$-summarization}, in which adjacent duplicate objects in the summary sequence are collapsed to reduce the size of a summarized sequence.

%  $\mathit{f}$-summarzation preserves the sequential relationship since it retains the indices of objects in the original sequence onto the summary sequence. The $\mathit{f}$-summarization defined in Definition~\ref{def:fsummary} does not reduce the size of original sequences, which is vital in improving the efficiency of sequence analysis. Therefore, we define the notion of \textit{reduced $\mathit{f}$-summarization}, in which adjacent duplicate objects in the summary sequence are collapsed to reduce size of summarized sequence.

\begin{definition} %(\textbf{\textit{Reduced $\mathit{f}$-Summarization}})
A reduced $\mathit{f}$-summarization of a sequence $\mathbf{p}$ on $\mathscr{O}$ is defined as a sequence $\mathbf{s}$ on $\mathscr{S}$, denoted as $\mathbf{s} = \mathit{f}^*(\mathbf{p})$, where each object $p \in \mathbf{p}$ is replaced by its $\mathit{f}$-based mapping $s = \mathit{f}(p)$ in $\mathbf{s}$ and, $\forall p_{i}, p_{i+1} \in \mathbf{p}, 1 \leq i \leq |\mathbf{p}|-1$, if $p_{i} = p_{i+1}$, then $\iota_{\mathbf{s}}(p_i) = \iota_{\mathbf{p}}(p_{i+1})$.
\label{def:reducedfsummary}
\end{definition}

\begin{theorem} %(\textbf{\textit{Sequential Preservation of Reduced \\ $\mathit{f}$-Summarization}})
A reduced $\mathit{f}$-summarization is sequence preserving.
\end{theorem}

\begin{proof}
    Given an original sequence $\mathbf{p} = (p_1, p_2, ..., p_m)$ on $\mathscr{O}$, let us denote $\mathbf{p}' = (p_1', p_2', ..., p_n')$ as a sequence on $\mathscr{S}$ and is the reduced $\mathit{f}$-summarization of $\mathbf{p}$. Elements in $\mathbf{p}'$ can also be described as follow: $p_1' = p_1$ and $p_i'=\mathit{f}(min_{\iota_{\mathbf{p}}(\cdot)} \{p_j: p_j \in \mathbf{p}, j \geq i, p_j \neq p_{i-1}\})$, for $1 < i \leq m$ (i.e., $p_i'$ is the $\mathit{f}$ mapping of the first non-duplicate element since $p_{i-1}$).
  
    Let us consider $p_i$ and $p_j \in \mathbf{p}$, $1 \leq i < j \leq m$. There are three possibilities:
    \begin{itemize}
    \item $p_i = p_j$ and $p_k = p_i (\forall k: i < k < j)$: In this case, we have $\iota_{\mathbf{p}'}(p_j) = \iota_{\mathbf{p}'}(p_i)$ and $\iota_{\mathbf{p}'}(p_k) = \iota_{\mathbf{p}'}(p_i)$, $\forall k: i < k < j$.
    \item $p_i = p_j$ and $\exists k: i < k < j, p_k \neq p_i$: In this case, we have $\iota_{\mathbf{p}'}(p_i) < \iota_{\mathbf{p}'}(p_k)$ and $\iota_{\mathbf{p}'}(p_k) = \iota_{\mathbf{p}'}(p_j)$. As a result, $\iota_{\mathbf{p}'}(p_i) < \iota_{\mathbf{p}'}(p_j)$.
    \item $p_i \neq p_j$: Since $1 \leq i < j \leq m$, we have  $\iota_{\mathbf{p}'}(p_i) < \iota_{\mathbf{p}'}(p_j)$ according to the above definition of $\mathbf{p}'$.
    \end{itemize}
  
    In all of the above cases, $\iota_{\mathbf{p}'}(p_i) \leq \iota_{\mathbf{p}'}(p_j)$, and thus, $\mathbf{p}'$ preserves the sequential relationship between elements in $\mathbf{p}$.
\end{proof}

\section{Summarization}
\label{sec:summarization}

%Having defined some basic notions of summarization,
In this section, we formally present our proposed summarization schemes\footnote{\scriptsize Unless explicitly stated, in the remaining sections, a summarization will always refer to reduced summarization.}.% on multi-dimensional sequences.

%After providing general definition of summarization in Section~\ref{sec:summarization}, in this section, we formally present our proposed summarization schemes on multi-dimensional sequences.

\subsection{Attribute-based Summarization}

To incorporate the multidimensional attributes of a sequence's data items, we first define the notion of \textit{attributes compatible mapping} that leverages a data item's attributes as a summarization criteria:

\begin{definition} %(\textbf{\textit{Attributes Compatible Mapping}})
Given a multidimensional set $\mathscr{O}=\langle\mathbb{O}, \mathbb{A}\rangle$ and a set of attributes $A \subseteq \mathbb{A}$, a mapping $\mathit{f}$ is defined as an $A$-compatible mapping if: $\forall p, q \in \mathscr{O}$, $\mathit{f}(p) = \mathit{f}(q)$ if and only if $\mathcal{A}_k(p) = \mathcal{A}_k(q), \forall \mathcal{A}_k \in A$.
\label{def:attr_mapping}
\end{definition}

Next, we define attribute-based summarization based on the definition of attributes compatible mapping:

\begin{definition} %(\textbf{\textit{Attribute-based Summarization}})
Given multidimensional set $\mathscr{O}=\langle\mathbb{O}, \mathbb{A}\rangle$ and a set of attributes $A \subset \mathbb{A}$, an $A$-based summarization is defined as a reduced $\mathit{f}$-summarization where the mapping $\mathit{f}$ is an $A$-compatible mapping on $\mathscr{O}$.
\label{def:attr_summary}
\end{definition}

% \begin{figure}
%   		\centering
%   		\begin{subfigure}[t]{0.2\textwidth}
%   			\centering
%         	\includegraphics[width=\textwidth]{activity-based.pdf}
%         	\caption{\scriptsize{\textit{Activity}-based}}
%         	\label{fig:example_summary_activity}
%     	\end{subfigure}
%     	\begin{subfigure}[t]{0.2\textwidth}
%     	    \centering
%         	\includegraphics[width=0.5\textwidth]{sector-based.pdf}
%         	\caption{\textit{Sector}-based}
%         	\label{fig:example_summary_sector}
%     	\end{subfigure}
%     	\begin{subfigure}[b]{0.2\textwidth}
%     		\centering
%         	\includegraphics[width=0.5\textwidth]{responsible-based.pdf}
%         	\caption{\textit{Responsible}-based}
%         	\label{fig:example_summary_responsible}
%     	\end{subfigure}
% 		\caption{Different forms of attribute-based summarization of the trace in Example~\ref{fig:motivating_example2}.}
% 		\label{fig:example_summary}
% \end{figure}

\begin{figure}[htp]
	\centering
	\begin{minipage}[tp]{0.30\textwidth}
		\centering
		\includegraphics[width=1\textwidth]{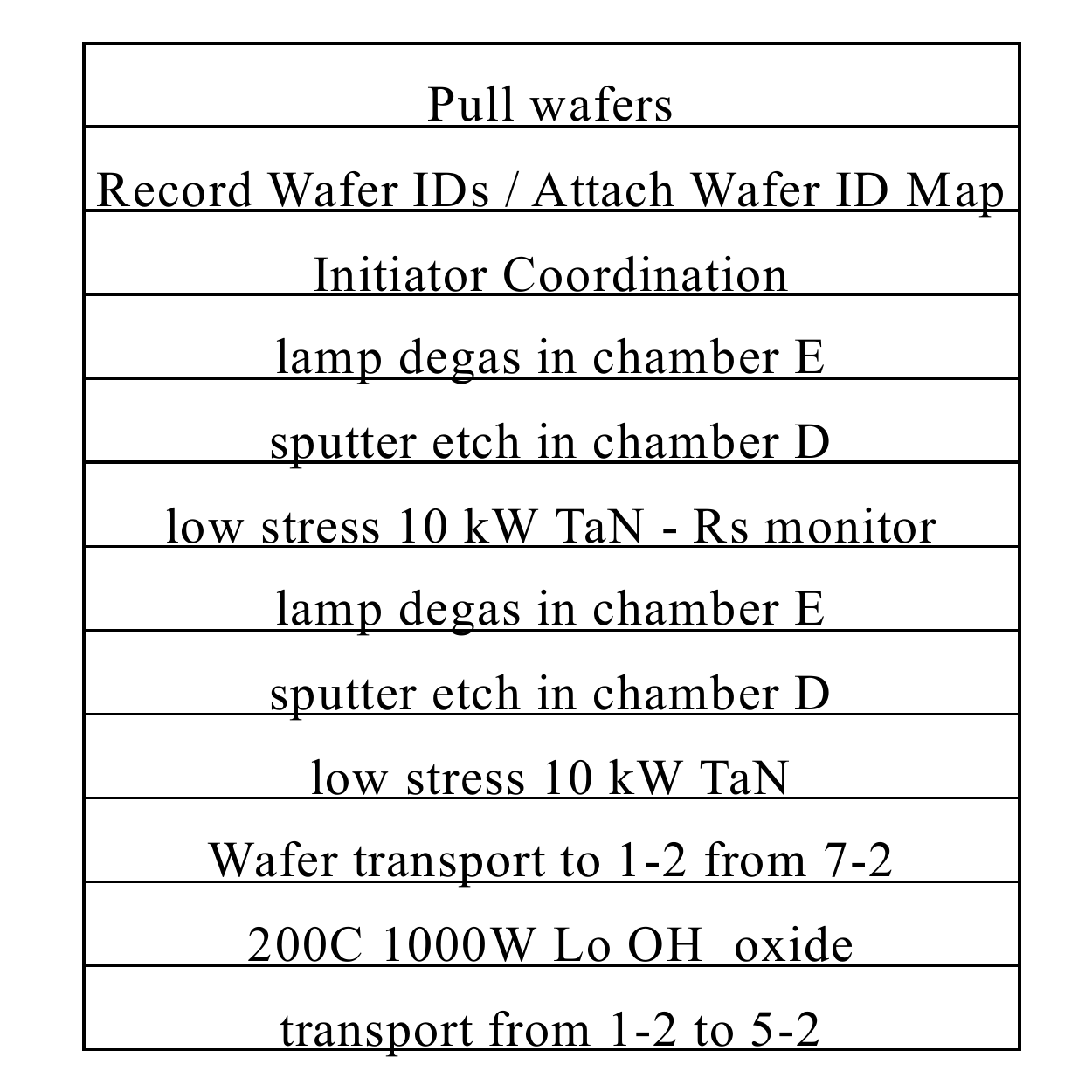}\subcaption{{\scriptsize \textit{Activity}-based}}\label{fig:example_summary_activity}
	\end{minipage}
	~
	\begin{minipage}[tp]{0.30\textwidth}
		\centering
		\includegraphics[width=0.5\textwidth]{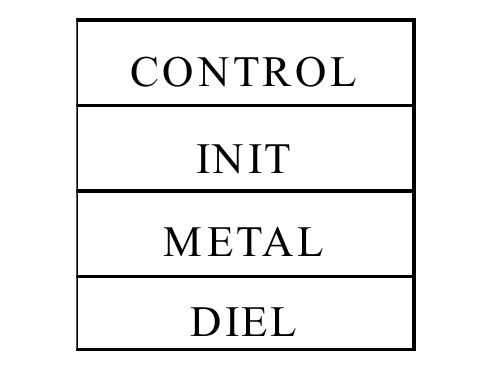} \subcaption{{\scriptsize\textit{Sector}-based}} \label{fig:example_summary_sector}
		\includegraphics[width=0.5\textwidth]{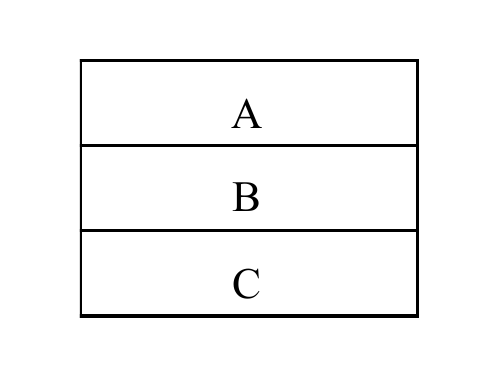}\subcaption{{{\scriptsize \textit{Responsible}-based}}}\label{fig:example_summary_responsible}
	\end{minipage}
	\caption{Different forms of attribute-based summarization of the trace in Example~\ref{fig:motivating_example2}.}
	\label{fig:example_summary}
	\vspace{-0.2in}
  \end{figure}

% Figure~\ref{fig:example_summary} shows examples of attribute compatible summarizations of the sample trace from Figure~\ref{fig:motivating_example2} (i.e., in this case, we use $Sector$ attribute as the summarization criteria) in its non-reduced form (i.e., Figure~\ref{fig:example_summary}a) and the corresponding reduced form (i.e., Figure~\ref{fig:example_summary}b).

% Although attribute compatible summarization provides an intuitive way for users to choose attributes as summarization criteria and produces summaries that are easy to interpret, it does not however give users much control over the \textit{resolution} of the summarized sequences. By \textit{resolution}, we refer to the number of dimensions of the summary space (i.e., $|\mathscr{S}|$). For example, assume that the sample trace in Example~\ref{fig:motivating_example2} is the only trace we have, then $Sector$ summary space has a resolution of 4 (i.e., {\it CONTROL}, {\it INIT}, {\it METAL}, and {\it DIEL}), while $Responsible$ summary space has resolution of 3 (i.e., $A$, $B$, and $C$). Since the resolution affects the size of sequences on summary space when using reduced summarization (e.g., the length of Example~\ref{fig:motivating_example2} trace on $Sector$ and $Responsible$ space is 4 and 3, respectively), users are not able to control the efficiency of the analysis on the summary space. Next, we introduce another summarization scheme that offers better control over the resolution of summaries, while still capturing semantic grouping of original dimensions.

Attribute compatible summarization provides an intuitive way for users to choose attributes as a summarization criteria and produces summaries that are easy to interpret. It does not give users control over the average length of summarized sequences, which we refer to as \textit{resolution}. This is because attribute values are static and already defined with the original data. Figure~\ref{fig:example_summary} shows examples of different attribute-based summarizations of the trace in Example~\ref{fig:motivating_example2}:  \textit{Activity}-based (Figure~\ref{fig:example_summary_activity}), \textit{Sector}-based (Figure~\ref{fig:example_summary_sector}), and \textit{Responsible}-based (Figure~\ref{fig:example_summary_responsible}). The \textit{Activity}-based summary the has biggest resolution among the examples, while the \textit{Responsible}-based summary has smallest resolution (i.e., the most compact summary).

Since longer summarized sequences reduce the efficiency of sequence data analysis, and attribute-based summarization offers users little flexibility in controlling that efficiency, it would be desirable if users are empowered to make the trade-off between efficiency and accuracy of data analysis, especially when dealing with large data or data of high complexity. For example, in a sequence similarity search application, users might decide to tolerate a certain level of false positives in the results (e.g., 0.9 false positive rate) to trade-off for faster response (e.g., results are returned within 5 seconds). To address this issue, we introduce a novel summarization scheme that offers more flexibility and better control over the resolution of summaries, while still capturing semantic and sequential relationships of the original data as with attribute-based summarization.

\begin{figure}[htp]
	\centering
	\includegraphics[width=0.5\linewidth]{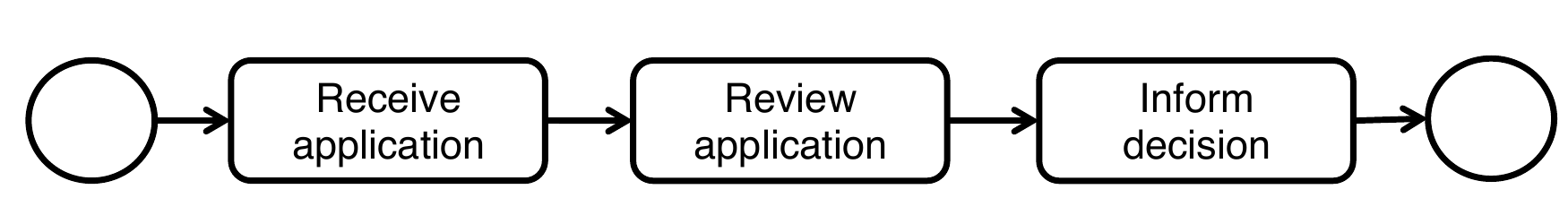}
  \caption{Topic-based representation of the process in Example 1.}
  \label{fig:topic_based_process}
  \vspace{-0.3in}
\end{figure}

\subsection{Topic-based Summarization}
We are motivated by the observation that business processes can often be represented by higher-level process models of fewer dimensions. Figure~\ref{fig:topic_based_process} shows an example of a more abstract process model of the one in Figure~\ref{fig:motivating_example1}. Each activity in Figure~\ref{fig:topic_based_process} corresponds to multiple activities in Figure~\ref{fig:motivating_example1}. We propose a \textit{topic-based summarization} technique that captures the many-to-one mapping from the original sequences to a \textit{topic-based summarization} of fewer dimensions, where each \textit{topic} is an abstract representation of a set of original dimensions. Since the topics are implicit from the original representation of sequences, we first perform dimensionality reduction on the original sequences to transform the original dimensions to topics. Then, we define the notion of topic-based summarization using the new representation.

%\subsubsection{Dimensionality Reduction on Sequences}

Before applying dimension reduction techniques to the original sequences, it is important to have an appropriate data representation for sequences. We begin by selecting an attribute of the original sequences and transform multidimensional sequences to the appropriate attribute-based summarization. It is often intuitive to pick the attribute with the most number of dimensions as this attribute likely captures the most essential information about the objects in the original multidimensional set. For example, in Example~\ref{fig:motivating_example2}, \textit{Activity} is the attribute with the most number of dimensions and it is also the base attribute to represent sequences, while other attributes, such as \textit{Sector} and \textit{Responsible}, provide supporting information for \textit{Activity}.

We then represent each sequence $\mathbf{p}$ as a numeric vector $(\vartheta_1, \vartheta_2, ..., \vartheta_{|\mathcal{A}^*|})$, where $\mathcal{A}^*$ is the base attribute set that sequences are transformed to in the first step and $|\mathcal{A}^*|$ is the number of dimensions on $\mathcal{A}^*$. We measure $\vartheta_i$ for $\mathbf{p}$ in a way that captures both the local importance of each dimension and its specificity to a sequence. To capture the local importance, we use the frequency of the $i$-th dimension in $\mathbf{p}$, denoted as $\mathtt{tf}^{i}_{\mathbf{p}}$, that is defined by the number of items in $\mathbf{p}$ whose values equal the $i$-th dimension of $\mathcal{A}^*$, denoted as $a_{i}$. To capture the specificity, we use the popularity of a dimension across all sequences: $\mathtt{df}_{i} = |\{\mathbf{p} \in \mathbb{S} | a_{i} \in \mathbf{p}\}|$, where $\mathbb{S}$ is the set of all sequences. Intuitively, the higher $\mathtt{df}_{i}$ is, the more popular the $i$-th dimension is and thus, the less specificity it is to a sequence. The formulation of $\vartheta_i$ is as follows:

\begin{equation}
\scriptsize
\vartheta_{i} =
	\begin{cases}
		(1 + log(\mathtt{tf}^{i}_{\mathbf{p}})) \times log(\frac{|\mathbb{S}|}{\mathtt{df}_{i}}) & \text{if $a_{i} \in \mathbf{p}$}\\
		0 & \text{otherwise}
	\end{cases}
\label{eq:dimred}
\end{equation}

After representing sequences as vectors, the set of sequences $\mathbb{S}$ can be represented as a matrix $\mathbf{M}$, whose size is $|\mathbb{S}| \times |\mathcal{A}^*|$ where each row corresponds to a vector representation of a sequence in $\mathbb{S}$. With this matrix representation, we can apply off-the-shelf dimension reduction techniques on $\mathbf{M}$, such as non-negative matrix factorization (NMF), principle component analysis (PCA), or singular value decomposition (SVD), among others.  The results of these techniques can be presented as two matrices $\mathbf{M}'$ and $\mathbf{W}$. $\mathbf{M}$, whose size equals $|\mathbb{S}| \times k$ with $k$ being the number of new dimensions (i.e., $k = |\mathscr{S}|$), represents the original sequences on the summary space. $\mathbf{W}$, whose size equals $|\mathscr{O}| \times k$, represents the original dimensions on the new dimensions, or topics (i.e., each row is a vector representing the distribution of an original dimension over the set of new dimensions).

%\subsubsection{Topic-based Summarization from Dimensionality Reduction Results}

Based on the results of dimensionality reduction, we now need to produce a many-to-one mapping from the original dimensions to topics. Two dimensions $a_{i}, a_{j}$ in the original space are likely to be in the same topic if their corresponding vectors in $\mathbf{W}$ have high similarity (e.g., using Cosine similarity). In addition, $a_{i}$ and $a_{j}$ are likely to be in the same topic if they frequently appear next to each other in a sequence (i.e., they represent two closely related activities in the underlining process model). From these insights, we model the problem of finding an optimal many-to-one mapping from the original dimensions to topics as a constrained optimization problem:

\begin{equation}
\scriptsize
\begin{aligned}
& \underset{\mathit{f}}{\text{argmax}}
& & \lambda \cdot \sum_{\mathit{f}(a_i)=\mathit{f}(a_j)} \theta(a_i, a_j) + (1-\lambda) \cdot \sum_{(a_i, a_j)} \omega(a_i, a_j)\theta(a_i, a_j) \\
& \text{subject to}
& & \mathit{f}: \mathscr{O} \rightarrow \mathscr{S} \\
& & & \forall a_i, a_j \in \mathscr{O},  \text{if } \mathit{f}(a_i) \neq \mathit{f}(a_j), \text{then } a_i \neq a_j.\\
& & & |\mathscr{S}| = k.
\end{aligned}
\label{eq:topic_optimization}
\end{equation}

where $\theta(a_i, a_j)$ is the similarity between dimensions $a_i$ and $a_j$ based on their corresponding representation in $\mathbf{W}$, $\omega(a_i, a_j)$ is the number of times $a_i$ and $a_j$ appear next to each other in input sequence set $\mathbb{S}$, and $\lambda$ is used to bias towards similarity between dimensions or the number of adjacent appearances.

We now can formally define the notion of topic summarization as follows:

\begin{definition} (\textbf{\textit{$k$-Topic Summarization}})
A $k$-topic summarization of sequences from original multidimensional set $\mathscr{O}$ to a summary set $\mathscr{S}$ is defined as a reduced $f$-summarization, where the mapping $\mathit{f}$ is the solution of the optimization problem defined in (\ref{eq:topic_optimization}).
\label{def:topic_summary}
\end{definition}

Finding an efficient k-topic summarization is the crux of the problem. We  say ``efficient'' as opposed to optimal because our k-topic summarization problem is NP-hard %({\it c.f.} Theorem \ref{thm:nphard}).
(a variant of the set partitioning problem).Thus we resort to a ``greedy'' heuristic approach. Our approach is similar to that of the agglomerative clustering algorithm. It starts with treating each original dimension as a singleton cluster and then successively merges pairs of dimensions that are closest to each other until all clusters have been merged into a single cluster that contains all dimensions. This step creates a hierarchy where each leaf node is a dimension and the root is the single cluster of the last merge. Because we want a partition of disjoint $k$ clusters as the new dimensions, the next step is to cut the hierarchy at some point to obtain the desirable number of clusters. To find the cut, we use a simple approach that is based on finding a minimum similarity threshold so that the distance between any two dimensions in the same cluster is no more than that threshold and there are at most $k$ clusters. %, and no more than $k$ clusters are formed.

% \begin{theorem} %(\textbf{\textit{NP-hardness of $k$-Topic Summarization}})
% 	\label{thm:nphard}
% 	The problem of finding a $k$-topic summarization of sequences on $\mathscr{O}$ is NP-hard.
% 	\end{theorem}

% 	\begin{proof}
% 		The optimization problem in Equation~\ref{eq:topic_optimization} is a variant of the set partitioning problem and finding a feasible solution for such a problem is NP-hard.
% 	\end{proof}

%Therefore, in this paper, we use a ``greedy'' heuristic approach. Our greedy approximation approach is similar to that of agglomerative clustering algorithm. Specifically, it starts with treating each original dimension as a singleton cluster and then successively merge pairs of dimensions that are closest to each other until all clusters have been merged into a single cluster that contains all dimensions. This step creates an hierarchy where each leaf node is a dimension and the root is the single cluster of the last merge. Because we want a partition of disjoint $k$ clusters as the new dimensions, the next step is to cut the hierarchy at some point to obtain the desirable number of clusters. To find the cut, we use a simple approach that is based on finding a minimum similarity threshold so that the distance between any two dimensions in the same cluster is no more than that threshold, and no more than $k$ clusters are formed.

Figure~\ref{fig:topic_summarization} outlines the process to generate $k$-topic summarization of sequences. There are two steps that require input from users: the number of topics (i.e., dimensions) on the summary space, and semantic labels for discovered topics. These inputs can be used by users to control the resolution of the summary space, as well as to integrate user expertise into the summarization (and thus, to the analysis tasks).

\begin{figure}[tp]
  	\centering
  	\includegraphics[width=1\linewidth]{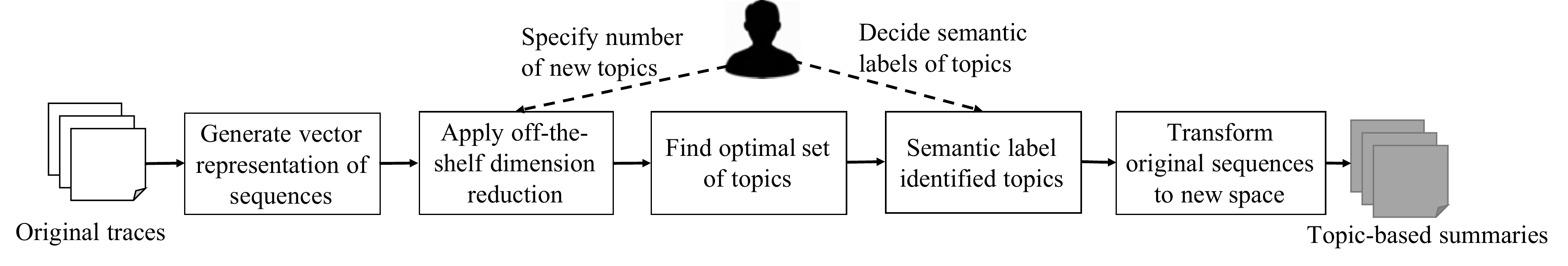}
	\caption{Topic summarization procedure.}
	\label{fig:topic_summarization}
	\vspace{-0.2in}
\end{figure}

\section{Error Model for Edit-Distance on Summaries}
\label{sec:errormodel}

%After presenting various summarization schemes, 
We seek to answer the question of how to relate approximate results of analysis tasks on the summary space to those on the original space. Since similarity measure is an important operator in a lot of analysis tasks, such as similarity search and traces clustering, we focus on the relationship between the similarity of sequences on the summary space with that on the original space under edit-distance constraint: $\mathtt{ed}(\mathbf{p}, \mathbf{q})$ \& $\mathtt{ed}(\mathit{f}(\mathbf{p}), \mathit{f}(\mathbf{q}))$, where $\mathtt{ed}$ is the edit-distance function and $\mathit{f}$ is a summarization function. We select Edit-distance as the similarity measure because it captures both the structural similarity (i.e., whether two sequences consist of data items in similar order) and content-based similarity (i.e., whether two sequences share similar set of data items) between sequences. Furthermore, Edit-distance's results, presented as a chain of edit operators to transform a sequence to the other, can be easily interpreted by users, which makes it widely popular in practice.

In terms of the relationship between $\mathtt{ed}(\mathbf{p}, \mathbf{q})$ and $\mathtt{ed}(\mathbf{\mathit{f}(\mathbf{p})}, \mathbf{\mathit{f}(\mathbf{q})})$, we are interested in the contractive  and proximity preservation properties.

\begin{definition} %(\textbf{\textit{Contractive Property}})
Given a summarization $\mathit{f}$, we said that the edit-distance measure satisfies the contractive property on $\mathit{f}$ if $\mathtt{ed}(\mathbf{p}, \mathbf{q}) \geq \mathtt{ed}(\mathit{f}(\mathbf{p}), \mathit{f}(\mathbf{q})), \forall \mathbf{p}, \mathbf{q}$.
\end{definition}

The contractive property is particularly important for applications such as similarity search, because it guarantees that performing edit-distance based similarity search on the summary space using $\mathit{f}$ will yield results with 100\% recall~\cite{hjaltason2003properties}\cite{papapetrou2009reference}. Specifically, given a query sequence $\mathbf{p}$ and an edit-distance threshold $\chi$, the similarity search task needs to find all sequences in the sequence set $\mathbb{S}$ that have edit-distance with $\mathbf{p}$ smaller or equal than $\chi$: $\mathbb{S}^{*} = \{\mathbf{q} \in \mathbb{S} | \mathtt{ed}(\mathbf{p}, \mathbf{q}) \leq \chi\}$. If the contractive property holds for a summarization $\mathit{f}$, we can avoid expensive calculation of edit-distance on the original space by finding all sequences $\mathbf{q}$ that satisfy the threshold $\chi$ on the summary space: $\bar{\mathbb{S}} = \{\mathbf{q} \in \mathbb{S} | \mathtt{ed}(\mathit{f}(\mathbf{p}), \mathit{f}(\mathbf{q})) \leq \chi\}$. Because if $\mathtt{ed}(\mathbf{p}, \mathbf{q}) \leq \chi$, then $\mathtt{ed}(\mathit{f}(\mathbf{p}), \mathit{f}(\mathbf{q})) \leq \chi$; we can guarantee that if $\mathbf{q} \in \mathbb{S}^{*}$, then $\mathbf{q} \in \bar{\mathbb{S}}$ (i.e., 100\% recall).

\begin{definition} %(\textbf{\textit{Proximity Preservation Property}})
Given a summarization $\mathit{f}$, we said that the edit-distance measure satisfies the proximity preservation property on $\mathit{f}$ if $\mathtt{ed}(\mathbf{p}, \mathbf{q}) \geq \mathtt{ed}(\mathbf{p}, \mathbf{r})$, then $\mathtt{ed}(\mathit{f}(\mathbf{p}), \mathit{f}(\mathbf{q})) \geq \mathtt{ed}(\mathit{f}(\mathbf{p}), \mathit{f}(\mathbf{r})), \forall \mathbf{p}, \mathbf{q}, \mathbf{r}$.
\end{definition}

The proximity preservation property is particularly important for applications such as traces clustering that group similar traces into the same cluster. This is because the proximity preservation property guarantees that traces that are similar in the original space are also similar in the summary space. Thus, the clustering results on the summary space will likely be similar to those on the original space.

%\subsubsection{Error Model for Edit-Distance on Summary Space}

%To address the questions regarding the relationship between $\mathtt{ed}(\mathbf{p}, \mathbf{q})$ and $\mathtt{ed}(\mathit{f}(\mathbf{p}), \mathit{f}(\mathbf{q}))$, we develop an error model for edit-distance measure on summary space.  
While the contractive property does not hold \textit{in general} for edit-distance between summarized sequences, we show that it holds under certain circumstances. The first of such circumstances is when the summarization $\mathit{f}$ is a \textit{non-reduced many-to-one}:
% We first show that, the contractive property does not hold \textit{in general} for edit-distance between summarized sequences:

% \begin{theorem}
% Contractive property does not hold in general for summarized sequences. 
% \label{theorem:contractivegeneral}
% \end{theorem}

% \begin{proof}
%    Check supplemental document for details.
% \end{proof}

%Next, we show that under certain circumstances, the contractive property holds for edit-distance between summarized sequences. The first of such circumstances is when the summarization $\mathit{f}$ is a \textit{non-reduced many-to-one}:

\begin{theorem}
If $\mathit{f}$ is a non-reduced many-to-one summarization on $\mathscr{O}$, as defined in Definition~\ref{def:fsummary}, then we have:  $\mathtt{ed}(\mathbf{p}, \mathbf{q}) \geq \mathtt{ed}(\mathit{f}(\mathbf{p}), \mathit{f}(\mathbf{q}))$, $\forall \mathbf{p}, \mathbf{q}$ on $\mathscr{O}$.
\label{theorem:nonreduced}
\end{theorem}

\begin{proof}
Let us assume that $\mathbf{p}=(p_1, p_2, ..., p_m)$, $ \mathbf{q} = (q_1, q_2, ...q_n)$. For compact representation, we denote $\mathtt{ed}(\mathbf{p}, \mathbf{q})$ as $\mathtt{ed}$ and $\mathtt{ed}(\mathit{f}(\mathbf{p}), \mathit{f}(\mathbf{q}))$ as $\mathtt{ed}'$.

As part of the recursive Wagner-Fischer algorithm to calculate edit-distance between two sequences $\mathbf{p}$ and $\mathbf{q}$, let us consider the step that involves comparing two data items $p_i \in \mathbf{p}$ and $q_j \in \mathbf{q}$ ($1 \leq i \leq m$, $1 \leq j \leq n$). If we denote the edit-distance at the current step as $\mathtt{ed}_{ij}$ and $\mathtt{ed}'_{ij}$ (for edit-distance on summary space), based on the recursive formula of the Wagner-Fischer algorithm, we have:

If $p_i = q_j$, then we have $\mathtt{ed}_{ij} = \mathtt{ed}_{i-1,j-1}$. Because of the many-to-one summarization $\mathit{f}$, $\mathit{f}(p_i) = \mathit{f}(q_j)$. Hence, $\mathtt{ed}'_{ij} = \mathtt{ed}'_{i-1,j-1}$. So, both $\mathtt{ed}_{ij}$ and $\mathtt{ed}'_{ij}$ \textit{do not require any edit cost} in this case.

If $p_i \neq q_j$, then we have $\mathtt{ed}_{ij} = min(\mathtt{ed}_{i-1,j}+1, \mathtt{ed}_{i,j-1}+1, \mathtt{ed}_{i-1,j-1}+1)$. Because of the many-to-one summarization $\mathit{f}$, we have either $\mathit{f}(p_i) = \mathit{f}(q_j)$ or $\mathit{f}(p_i) \neq \mathit{f}(q_j)$. Thus, $\mathtt{ed}'_{ij} = min(\mathtt{ed}'_{i-1,j}+1, \mathtt{ed}'_{i,j-1}+1, \mathtt{ed}'_{i-1,j-1}+1)$ if $\mathit{f}(p_i) \neq \mathit{f}(q_j)$ (i.e., one edit cost), or $\mathtt{ed}'_{ij} = \mathtt{ed}'_{i-1,j-1}$ if $\mathit{f}(p_i) = \mathit{f}(q_j)$ (i.e., no edit cost). So, in this case, $\mathtt{ed}_{ij}$ \textit{requires one edit cost}, while $\mathtt{ed}'_{ij}$ \textit{requires either one or zero edit cost}.

Therefore, we always have $\mathtt{eq}_{ij} \geq \mathtt{eq}'_{ij}, \forall i, j$. Since the values $\{\mathtt{eq}_{ij}\}$ and $\{\mathtt{eq}'_{ij}\}$ form the matrix used by recursive algorithm to calculate $\mathtt{ed}(\mathbf{p}, \mathbf{q})$ and $\mathtt{ed}(\mathit{f}(\mathbf{p}), \mathit{f}(\mathbf{q}))$ respectively, then we have $\mathtt{ed}(\mathbf{p}, \mathbf{q}) \geq \mathtt{ed}(\mathit{f}(\mathbf{p}), \mathit{f}(\mathbf{q}))$.
\end{proof}

Consider the case when $\mathit{f}$ is a \textit{reduced many-to-one summarization}.
% VM: Is it ok to remove this next line?
%Although we are not able to prove that the contractive property holds for any sequence pairs using summarization $\mathit{f}$
We are able to derive rules to indicate whether the contractive property holds or does not hold for edit-distance of a particular pair of sequences $\mathbf{p}, \mathbf{q}$ using summarization $\mathit{f}$:

\begin{theorem}
Given two sequences $\mathbf{p}, \mathbf{q}$ in the original space $\mathscr{O}$, if $\mathit{f}$ is a reduced many-to-one summarization on $\mathscr{O}$, as defined in Definition~\ref{def:reducedfsummary}, then:

\begin{itemize}
\item If $\Gamma_{\mathbf{p}, \mathbf{q}} \geq \Lambda_{\mathit{f}(\mathbf{p}), \mathit{f}(\mathbf{q})}$, then we have $\mathtt{ed}(\mathbf{p}, \mathbf{q}) \geq \mathtt{ed}(\mathit{f}(\mathbf{p}), \mathit{f}(\mathbf{q}))$; or edit-distance on summary space by $\mathit{f}$ \textit{\textbf{satisfies}} the contractive property.
\item If $\Gamma_{\mathit{f}(\mathbf{p}), \mathit{f}(\mathbf{q})} > \Lambda_{\mathbf{p}, \mathbf{q}}$, then we have $\mathtt{ed}(\mathbf{p}, \mathbf{q}) < \mathtt{ed}(\mathit{f}(\mathbf{p}), \mathit{f}(\mathbf{q}))$; or edit-distance on summary space by $\mathit{f}$ \textit{\textbf{does not}} satisfy the contractive property.
\end{itemize}
where $\Lambda_{\mathbf{p}, \mathbf{q}} = max(|\mathbf{p}|, |\mathbf{q}|)$ and $\Gamma_{\mathbf{p}, \mathbf{q}} = ||\mathbf{p}|-|\mathbf{q}||$, with $|\mathbf{p}|$ being the length of $\mathbf{p}$.
\label{theorem:reduced}
\end{theorem}

\begin{proof}
    This theorem can be easily proven by noticing that $\Lambda$ and $\Gamma$ in fact define the upper bound and lower bound on the edit-distance of a pair of sequences.

    The first rule is proven by using the chain rule of inequality: $\mathtt{ed}(\mathbf{p}, \mathbf{q}) \geq \Gamma_{\mathbf{p}, \mathbf{q}} \geq \Lambda_{\mathit{f}(\mathbf{p}), \mathit{f}(\mathbf{q})} \geq \mathtt{ed}(\mathit{f}(\mathbf{p}), \mathit{f}(\mathbf{q}))$.

    Similarly, for the second rule: $\mathtt{ed}(\mathit{f}(\mathbf{p}), \mathit{f}(\mathbf{q})) \geq \Gamma_{\mathit{f}(\mathbf{p}), \mathit{f}(\mathbf{q})} > \Lambda_{\mathbf{p}, \mathbf{q}} \geq \mathtt{ed}(\mathbf{p}, \mathbf{q})$.
\end{proof}

% \begin{proof}
% This theorem can be easily proven by noticing that $\Lambda$ and $\Gamma$ in fact define the upper bound and lower bound on the edit-distance of a pair of sequences.

% The first rule is proven by using the chain rule of inequality: $\mathtt{ed}(\mathbf{p}, \mathbf{q}) \geq \Gamma_{\mathbf{p}, \mathbf{q}} \geq \Lambda_{\mathit{f}(\mathbf{p}), \mathit{f}(\mathbf{q})} \geq \mathtt{ed}(\mathit{f}(\mathbf{p}), \mathit{f}(\mathbf{q}))$.

% Similarly, for the second rule: $\mathtt{ed}(\mathit{f}(\mathbf{p}), \mathit{f}(\mathbf{q})) \geq \Gamma_{\mathit{f}(\mathbf{p}), \mathit{f}(\mathbf{q})} > \Lambda_{\mathbf{p}, \mathbf{q}} \geq \mathtt{ed}(\mathbf{p}, \mathbf{q})$.
% \end{proof}

While Theorem~\ref{theorem:reduced} does not cover all cases, we show empirically that the number of sequence pairs whose edit-distances on reduced summarization violate the contractive property is very small. Thus, it has a high recall for similarity search task when using reduced many-to-one summarization.

For the \textit{proximity preservation property}, 
% VM: Is it ok to remove the next phrase?
%while we are not able to obtain a theoretical guarantee
we are able to show in our evaluation that the edit distance-based traces clustering results in the summary space have comparable accuracy, compared with those in the original space, while having better efficiency. This implies that the proximity relationship is well-preserved in the summary space under edit-distance constraint.

\section{Evaluation}
\label{sec:evaluation}

We demonstrate the utility of our \textsc{Summarized} framework by evaluating its effectiveness and efficiency on two analysis tasks: trace similarity search and traces clustering.

%In this section, we demonstrate the utility of our \textsc{Summarized} framework by presenting a thorough experimental evaluation of the effectiveness and efficiency of the different schemes we propose.
%\subsection{Evaluation Settings}

%{\noindent \textbf{\textit{Analysis tasks:}}} We evaluate the effectiveness and efficiency of the summarization schemes on two analysis tasks: trace similarity search and traces clustering.

{\noindent \textbf{\textit{Datasets:}}} We use three datasets from different domains: the $\mathtt{Lithography}$ dataset (596 traces with 1066 types of activities, each activity has multi-dimensional attributes) that contains traces generated from the executions of a semiconductor manufacturing process, the $\mathtt{BPIC}$ 2015 dataset (1199 traces with 289 activity types) that contains process traces of building permit applications, and a $\mathtt{BANK}$ dataset (2000 traces with 113 activity types) that consists of synthetically generated logs that represent a large bank transaction process\footnote{\scriptsize The $\mathtt{Lithography}$ dataset is provided by IBM and is private. The other datasets are available at \url{https://data.4tu.nl/repository/collection:all}.}. We run our evaluation on a computer with 16GB of RAM and a 2.7GHz quad-core Intel Core i7 CPU.

%\footnote{\scriptsize The $\mathtt{Lithography}$ dataset is provided by IBM and is confidential. The $\mathtt{BPIC}$ and $\mathtt{BANK}$ datasets are publicly available at \url{https://data.4tu.nl/repository/collection:event_logs}.}. We run our evaluation on a computer with 16GB of RAM and a 2.7GHz quad-core Intel Core i7 CPU.

{\noindent \textbf{\textit{Summarization schemes:}}} We compare results of analysis tasks in the summary space using our proposed summarization schemes (i.e., $\mathtt{Topic}$ and $\mathtt{Attribute}$), $\mathtt{Random}$ summarization%\footnote{\scriptsize Although $\mathit{Random}$-based summaries lack of the interpretability, as shown in \cite{chen2009mining}, a random summarization scheme on graphs, which sequences are special cases, can yield surprisingly good results.}%
, which randomly maps an original dimension to a new dimension in the summary space, and with the analysis results on the original space. Although $\mathit{Random}$-based summaries lack interpretability, as shown in \cite{chen2009mining}, a random summarization scheme on sequence graph can yield good results. We vary the number of dimensions $k$ in the summary space used by $\mathtt{Random}$ and $\mathtt{Topic}$ and vary the attributes used by $\mathtt{Attribute}$.

%\begin{figure}[tp]
%  \vspace{-0.1 in}
%    \centering
%    \includegraphics[width=.50\linewidth]{contractive_violations.pdf}
%  \caption{Number of sequence pairs in the $\mathtt{Lithography}$ dataset where edit-distance in the summary space violates the contractive property.}
%  \label{fig:contractive_violations}
%  \vspace{-0.2 in}
%\end{figure}

\begin{figure}[tp] \centering \scriptsize
\addtolength{\tabcolsep}{4pt}  
  \begin{tabular}{|r|c|c|c|c|c|c|}
  \hline
   & k=2 & k=5 & k=10 & k=20 & k=50 & k=100 \\
  \hline
  \hline
  Topic
  	& 0.000\%
  	& 0.003\%
  	& 0.006\%
  	& 0.007\%
  	& 0.010\%
  	& 0.014\% 
  	\\
  Random
  	& 0.002\%
  	& 0.010\%
  	& 0.007\%
  	& 0.021\%
  	& 0.027\%
  	& 0.033\%
  	\\
  \hline
  \end{tabular} 
\addtolength{\tabcolsep}{-4pt}  
\caption{Similarity false negatives: percentage of sequence pairs in the $\mathtt{Lithography}$ dataset where edit-distance in the summary space violates the contractive property. There are over 177,000 total sequence pairs in the dataset.}
\label{fig:contractive_violations}
\end{figure}

{\noindent \textbf{\textit{Evaluation metrics for the similarity search task:}}} The contractive property holds for most of the cases, as seen in Figure~\ref{fig:contractive_violations} which shows the percentage of sequence pairs in the $\mathtt{Lithography}$ dataset, out of over 177,000 pairs, whose edit-distances violate the contractive property in the summary space using $Topic$ and $Random$ summarization over different number of summary dimensions $k$. Since the recall rate is high, we focus on the false positive rate of the similarity search results. Given an edit distance threshold $\chi$, this metric tells us that, out of all sequence pairs that satisfy $\mathtt{ed}(\mathit{f}(\mathbf{p}), \mathit{f}(\mathbf{q})) \leq \chi$ on the summary space, how many of them actually satisfy the threshold in the original space: $\mathtt{ed}(\mathbf{p}, \mathbf{q}) \leq \chi$.

%\subsection{Effectiveness of Summarization Schemes on Similarity Search}

\smallbreak
%{\noindent \textbf{\textit{Effectiveness of summarization schemes on similarity search:}} To evaluate the effectiveness of different summarization schemes on the similarity search task, we vary the number of dimensions in the summary space used by $\mathtt{Random}$ and $\mathtt{Topic}$; and vary the attribute used by $\mathtt{Attribute}$ and compare the false positive results over different edit-distance thresholds.

%Figure~\ref{fig:fp_all} shows the results
{\noindent \textbf{\textit{Effectiveness of summarization schemes on similarity search:}}  Figure~\ref{fig:fp_all} shows the effectiveness of different summarization schemes on the similarity search task for the $\mathtt{Lithography}$, $\mathtt{BPIC}$, and $\mathtt{BANK}$ datasets\footnote{\scriptsize We only evaluate $\mathtt{Attribute}$ summarization on the $\mathtt{Lithography}$ dataset because this dataset's attributes provide better semantics compared with the ones in $\mathtt{BPIC}$ and $\mathtt{BANK}$.}. The y-axis reports the false positive results, while the x-axis corresponds to different edit-distance thresholds.  As expected (Figure~\ref{fig:litho_fp_random}, \ref{fig:litho_fp_topic}, \ref{fig:bpic_fp_random}, \ref{fig:bpic_fp_topic}, \ref{fig:bank_fp_topic}, \ref{fig:bank_fp_random}), the higher the number of dimensions in the summary space (denoted by $k$), the better the result (i.e., lower false positive rates). That is because, with more dimensions in the summary space, summaries of sequences more resemble the original sequences. Thus, there is little difference between edit-distances on the summary space and in the original space (hence, lower false positive rate).

\begin{figure}
	\centering
    \begin{subfigure}[t]{0.31\linewidth}
        \centering
  		\includegraphics[width=1\linewidth]{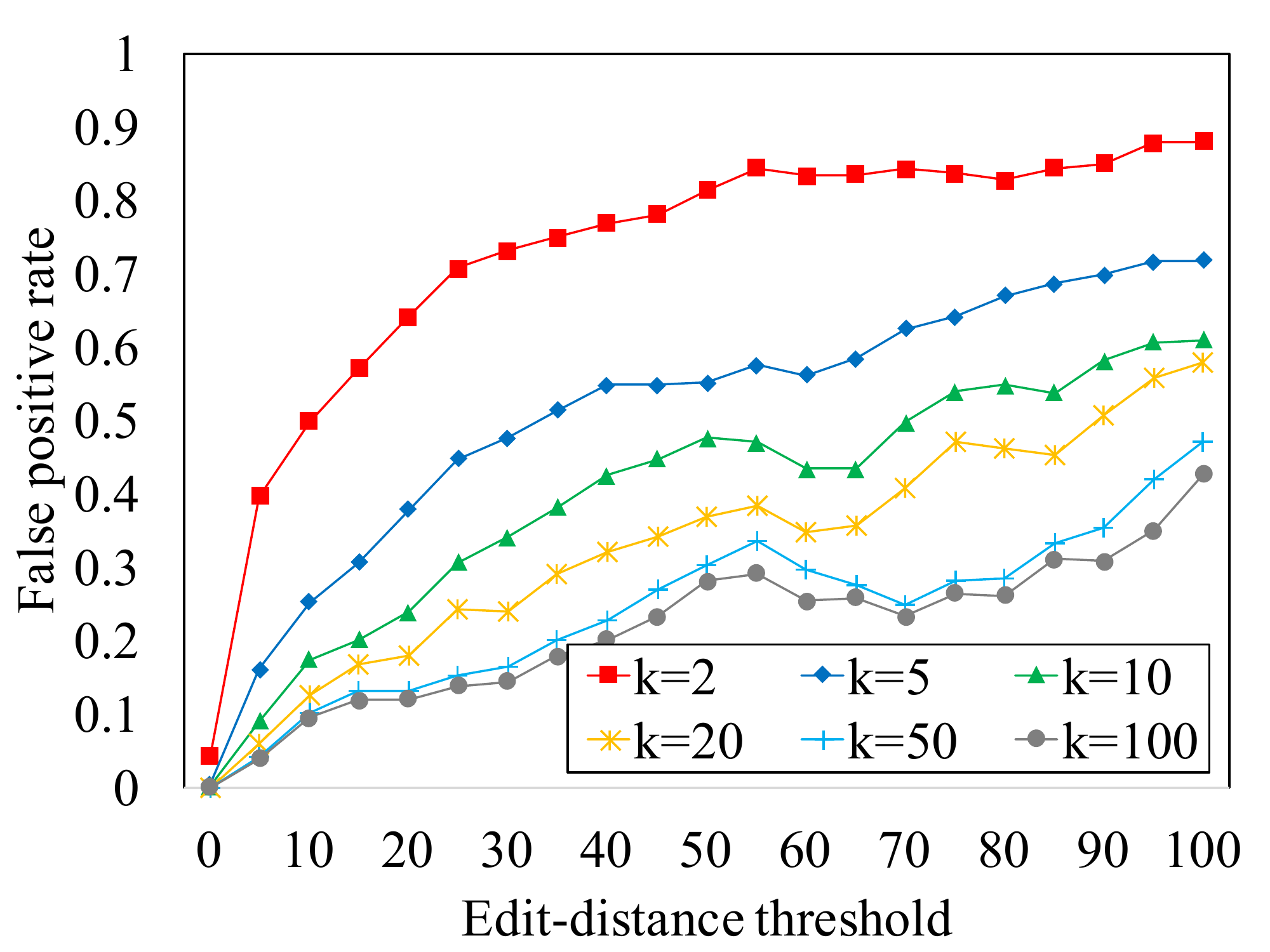}
		\caption{\scriptsize {Random ($\mathtt{Lithography}$)}}
		\label{fig:litho_fp_random}
    \end{subfigure}
    ~
    \begin{subfigure}[t]{0.31\linewidth}
        \centering
        \includegraphics[width=1\textwidth]{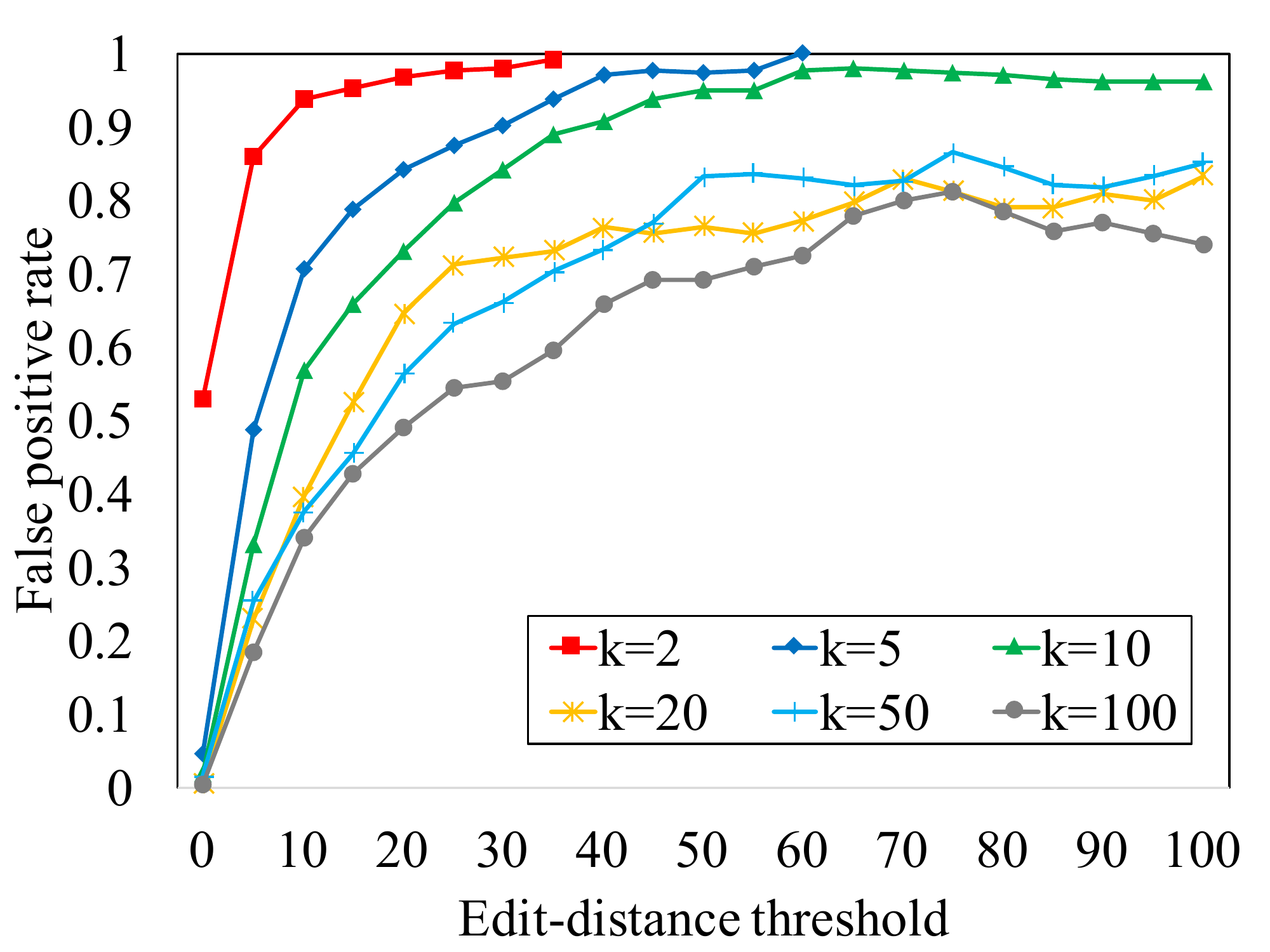}
        \caption{\scriptsize {Topic ($\mathtt{Lithography}$)}}
		\label{fig:litho_fp_topic}
    \end{subfigure}
    ~
    \begin{subfigure}[t]{0.31\linewidth}
        \centering
        \includegraphics[width=1\textwidth]{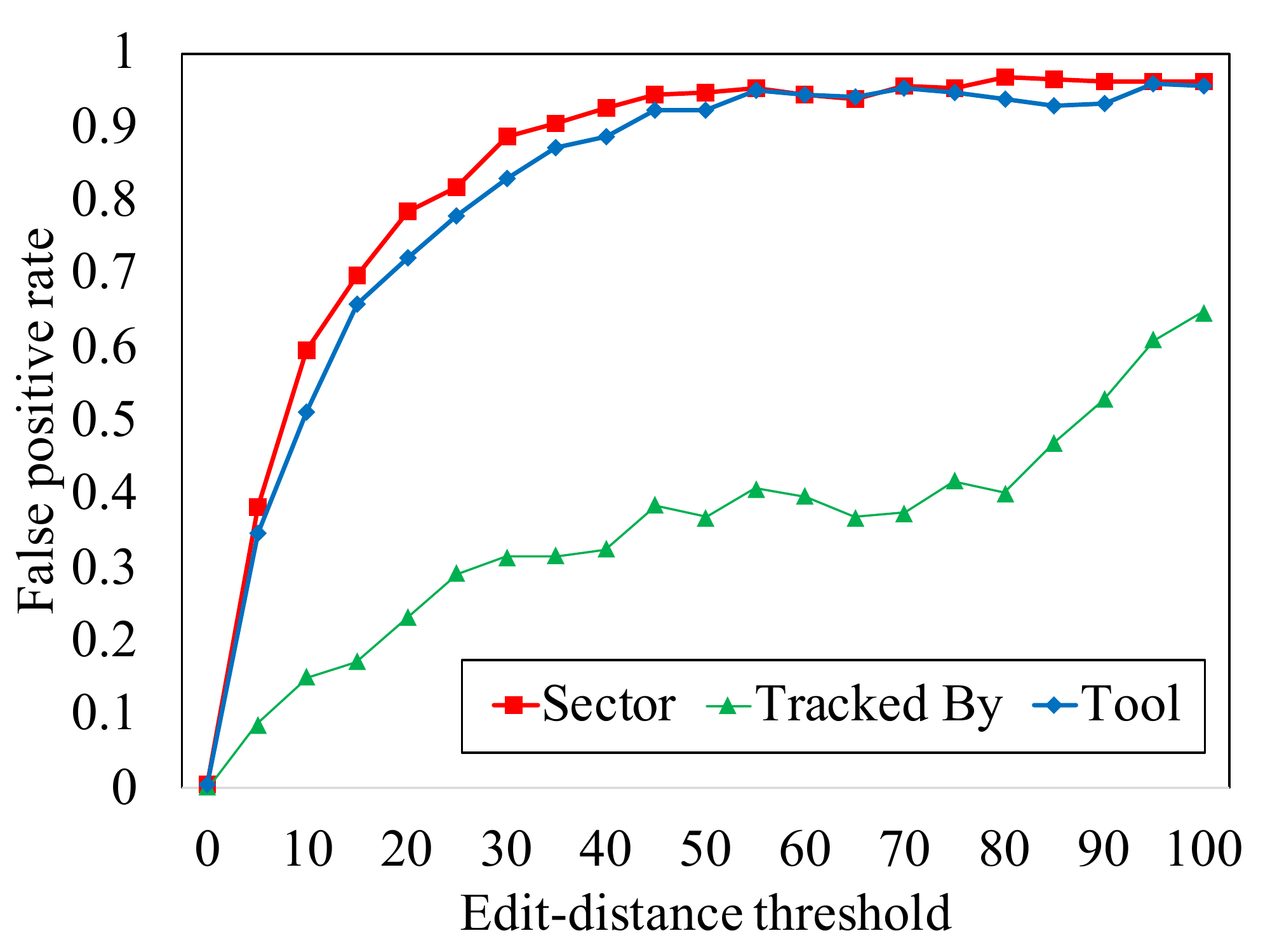}
        \caption{\scriptsize {Attribute ($\mathtt{Lithography}$)}}
		\label{fig:litho_fp_attribute}
    \end{subfigure}
    ~
    \begin{subfigure}[t]{0.23\linewidth}
        \centering
  		\includegraphics[width=1\linewidth]{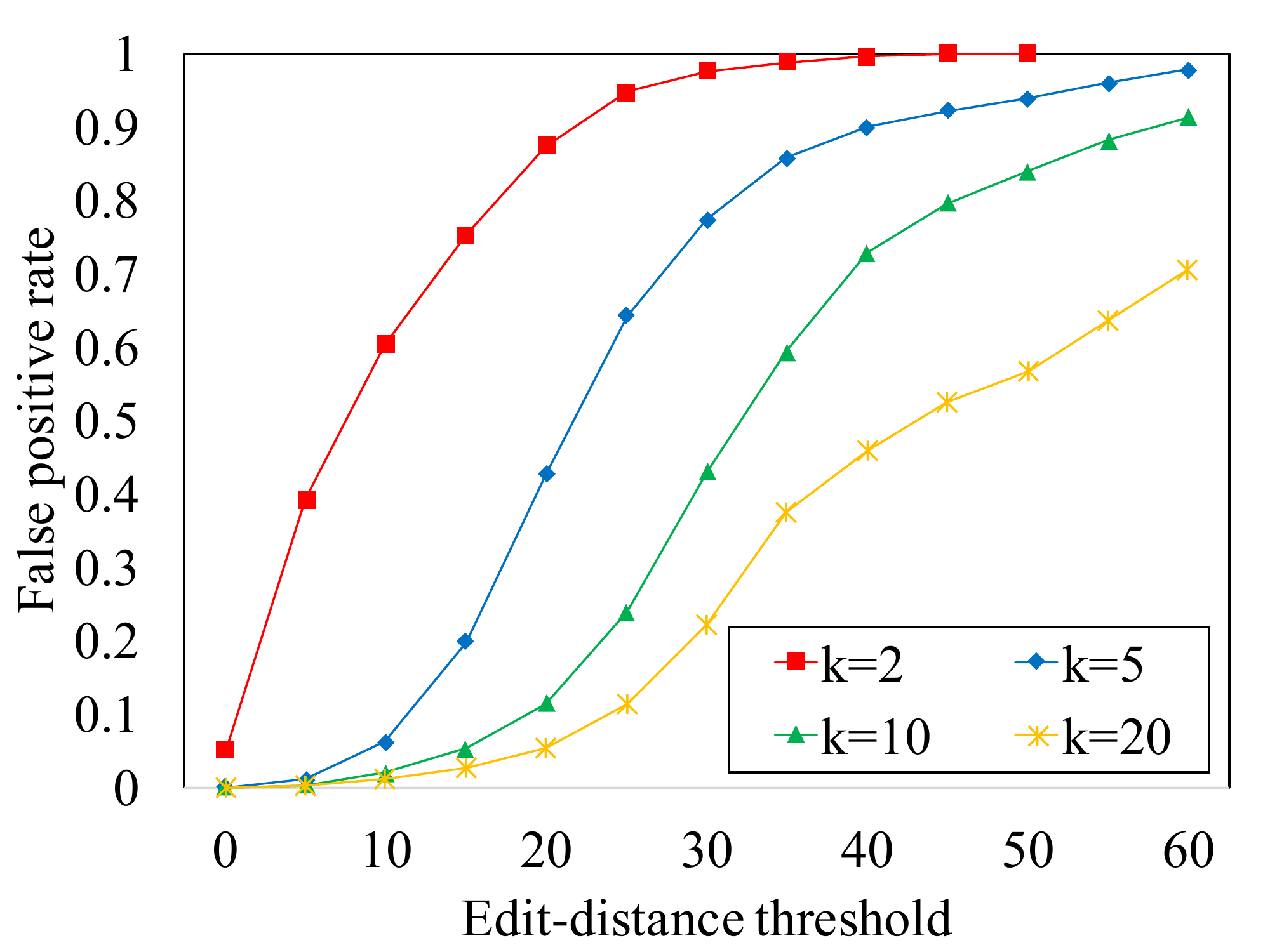}
		\caption{\scriptsize {Random ($\mathtt{BPIC}$)}}
		\label{fig:bpic_fp_random}
    \end{subfigure}
    ~
    \begin{subfigure}[t]{0.23\linewidth}
        \centering
        \includegraphics[width=1\textwidth]{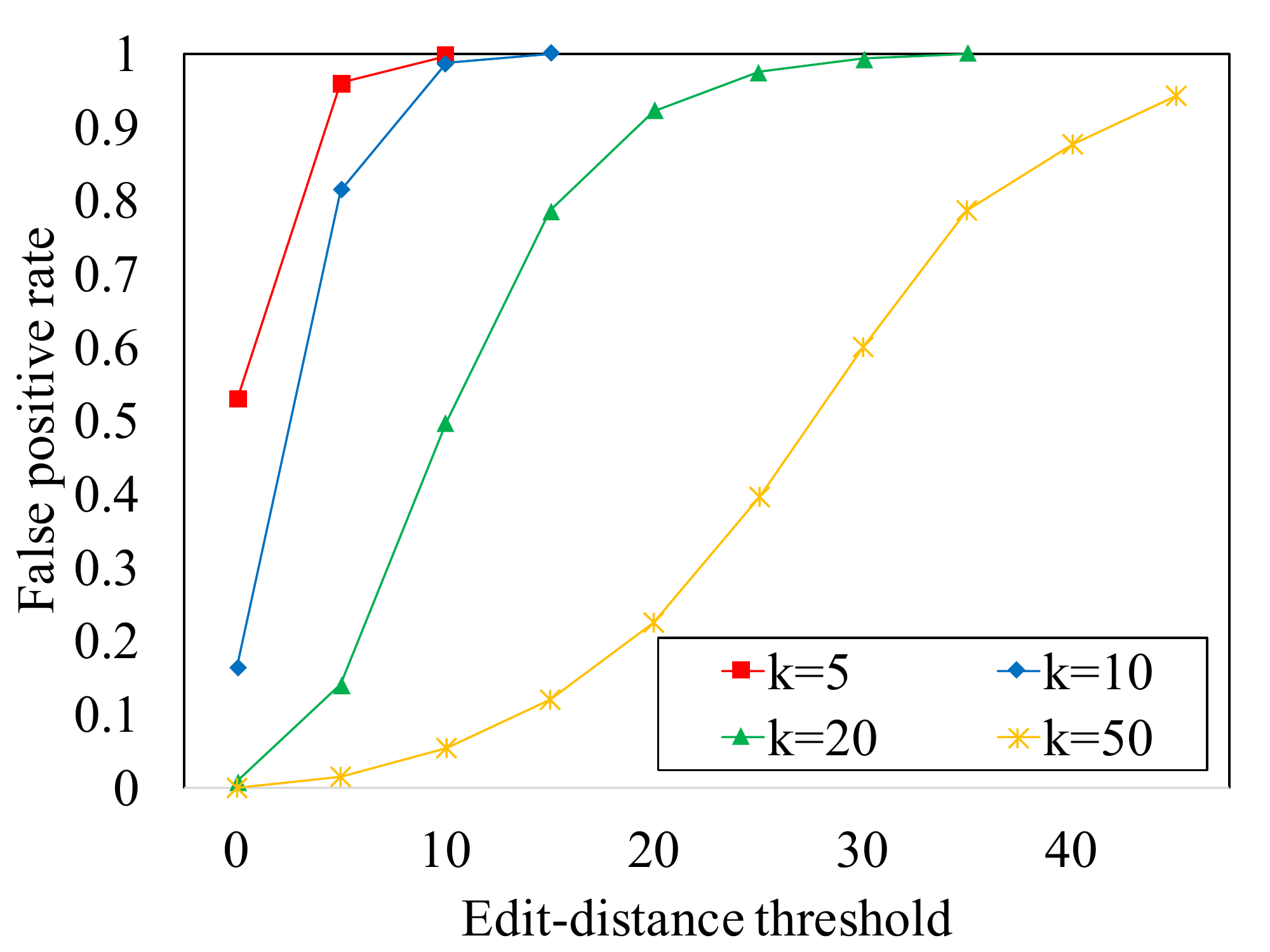}
        \caption{\scriptsize {Topic ($\mathtt{BPIC}$)}}
		\label{fig:bpic_fp_topic}
    \end{subfigure}
    ~
    \begin{subfigure}[t]{0.23\linewidth}
        \centering
        \includegraphics[width=1\linewidth]{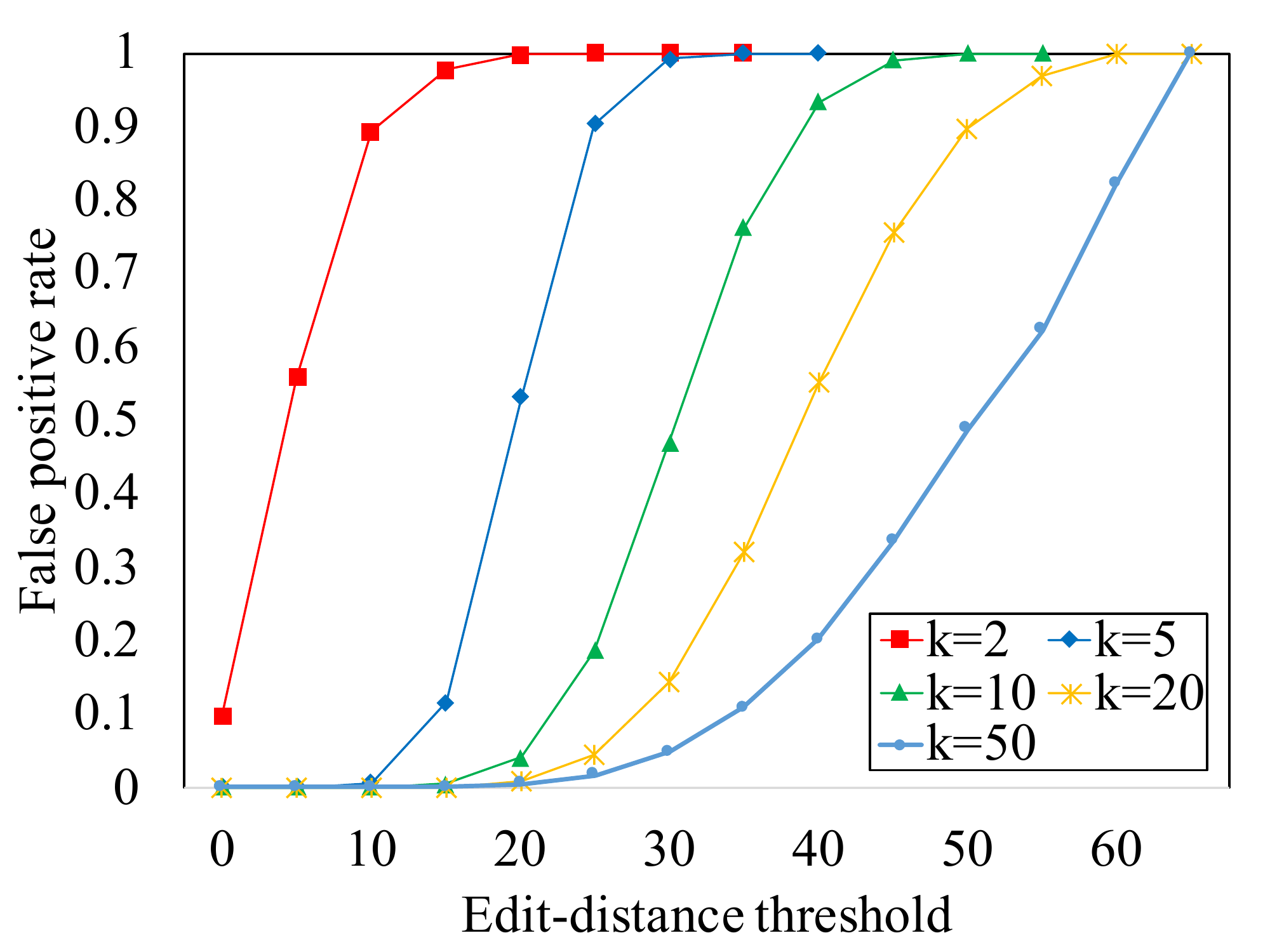}
      \caption{\scriptsize {Random ($\mathtt{BANK}$)}}
      \label{fig:bank_fp_random}
    \end{subfigure}
     ~
     \begin{subfigure}[t]{0.23\linewidth}
  	 	\centering
	   	\includegraphics[width=1\linewidth]{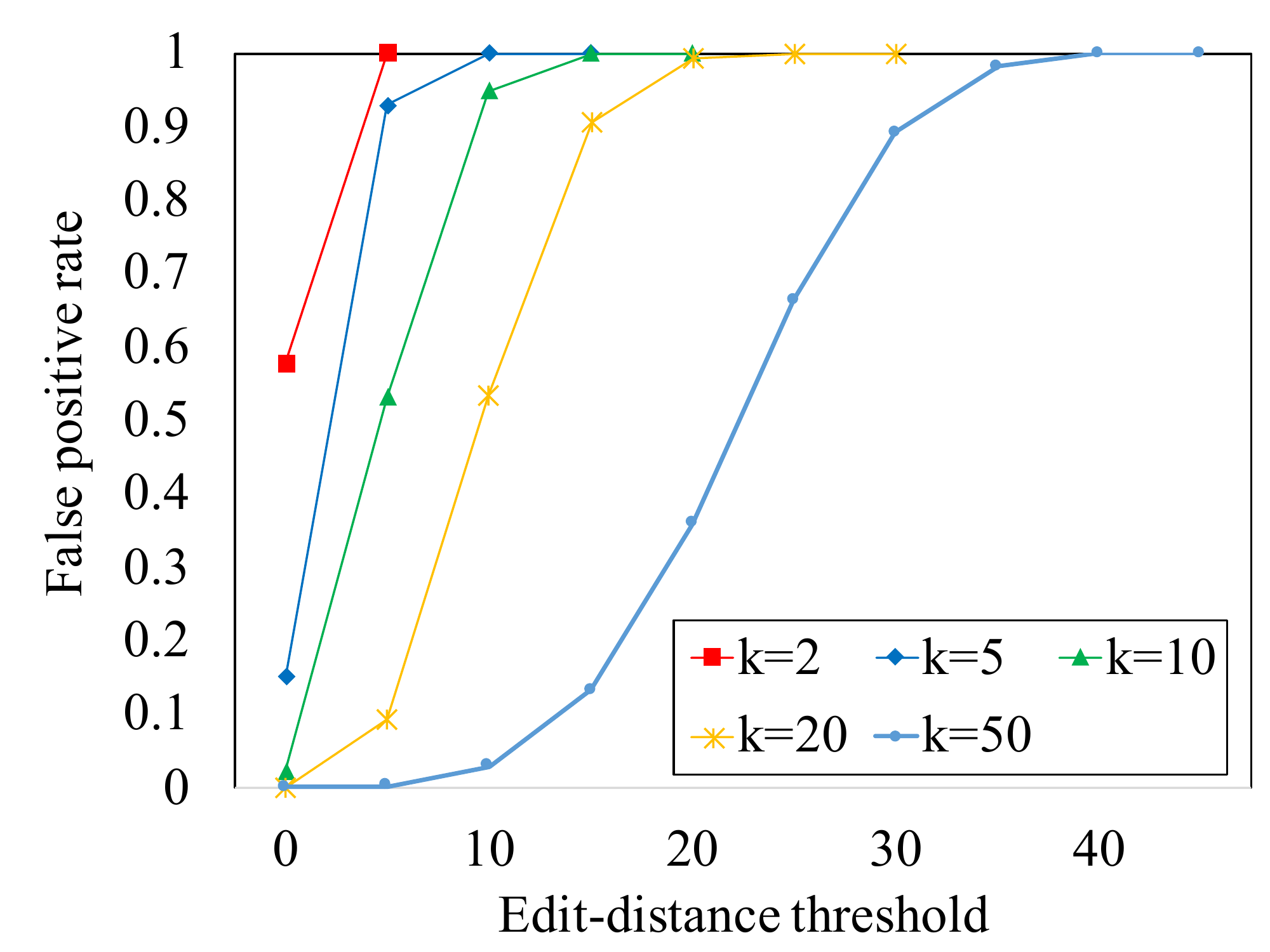}
	 	\caption{\scriptsize{Topic ($\mathtt{BANK}$)}}
	 	\label{fig:bank_fp_topic}
	 \end{subfigure}
	 \caption{False positive rates by different summarization schemes on similarity search task using the $\mathtt{Lithography}$, $\mathtt{BPIC}$, and $\mathtt{BANK}$ datasets.}
  \label{fig:fp_all}
  \vspace{-0.1in}
\end{figure}

When comparing the results of different summarization schemes on the same number of dimensions, $\mathtt{Random}$ outperforms $\mathtt{Topic}$ summarization (at the cost of interpretability of summaries and the efficiency, as we will show later). For $\mathtt{Attribute}$ (Figure~\ref{fig:litho_fp_attribute}), since we do not have control over the number of dimensions (since it depends on the attribute data), the quality of the results also depend on the chosen attribute. Specifically, the $Tracked By$\footnote{\scriptsize Three main attributes of an activity are used on the $\mathtt{Lithography}$ data: $Tracked By$ represents the person in charged of the activity; $Sector$ represents the area/department where the activity is taken, and $Tool$ represents the tool used to perform the activity.} attribute outperforms $Sector$ and $Tool$. This is in part because there are more dimensions on $Tracked By$'s summary space, and thus the summaries on the $Tracked By$ space more resemble the original sequences. $Sector$ and $Tool$ produce similar results, since similar $Tool$s are often used in the same $Sector$.

\smallbreak
{\noindent \textbf{\textit{Efficiency of summarization schemes on similarity calculation:}} To evaluate the efficiency of different summarization schemes, we vary the number of dimensions $k$ in the summary space and measure the time it takes to calculate the edit-distance similarity between all pairs of sequences. %\footnote{\scriptsize The results of this evaluation could be applied for both similarity search and traces clustering tasks, since both tasks involve a lot of similarity calculations.}.
Figure~\ref{fig:efficiency_all} highlights the results. 
For both $\mathtt{Random}$ and $\mathtt{Topic}$ summarizations\footnote{\scriptsize Again, since we could not control the number of dimensions of $Attribute$, we do not include it in this evaluation. However, $\mathtt{Attribute}$ produces similar efficiency results to the summarization schemes that share similar number of dimensions.}, the higher $k$ is, the longer it takes to calculate the edit-distances. This is expected because a higher $k$ results in longer sequences in the summary space, and thus it is more expensive to calculate the edit-distances. For similar values of $k$, $\mathtt{Topic}$ outperforms $\mathtt{Random}$, which verifies $\mathtt{Topic}$'s ability to capture the semantic relationship between the original dimensions, and thus significantly reduces the size of sequences in the summary space, as well as the processing time.  More importantly, even at different values of $k$ where we observed similar effectiveness of results by $\mathtt{Random}$ and $\mathtt{Topic}$ (e.g., $k=2$ with $\mathtt{Random}$ and $k=10$ with $\mathtt{Topic}$ on the $\mathtt{Lithography}$ dataset in Figure~\ref{fig:fp_all}), $\mathtt{Topic}$ is still much more efficient than $\mathtt{Random}$.

\begin{figure}[tp]
	\centering
	\begin{subfigure}[t]{0.32\linewidth}
        \centering
        \includegraphics[width=1\textwidth]{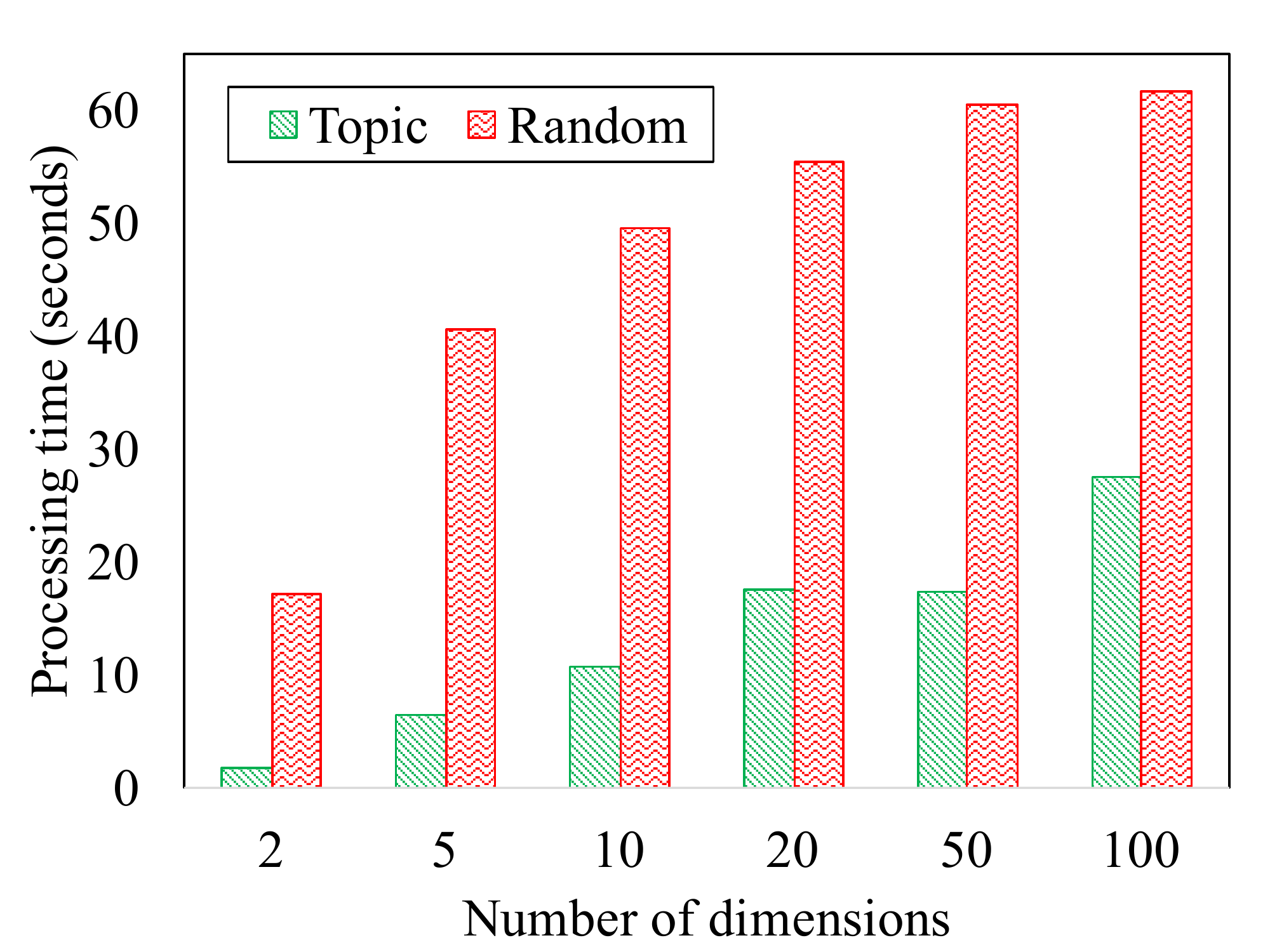}
        \caption{\scriptsize {($\mathtt{Lithography}$)}}
		\label{fig:litho_time}
    \end{subfigure}
    \begin{subfigure}[t]{0.32\linewidth}
        \centering
        \includegraphics[width=1\textwidth]{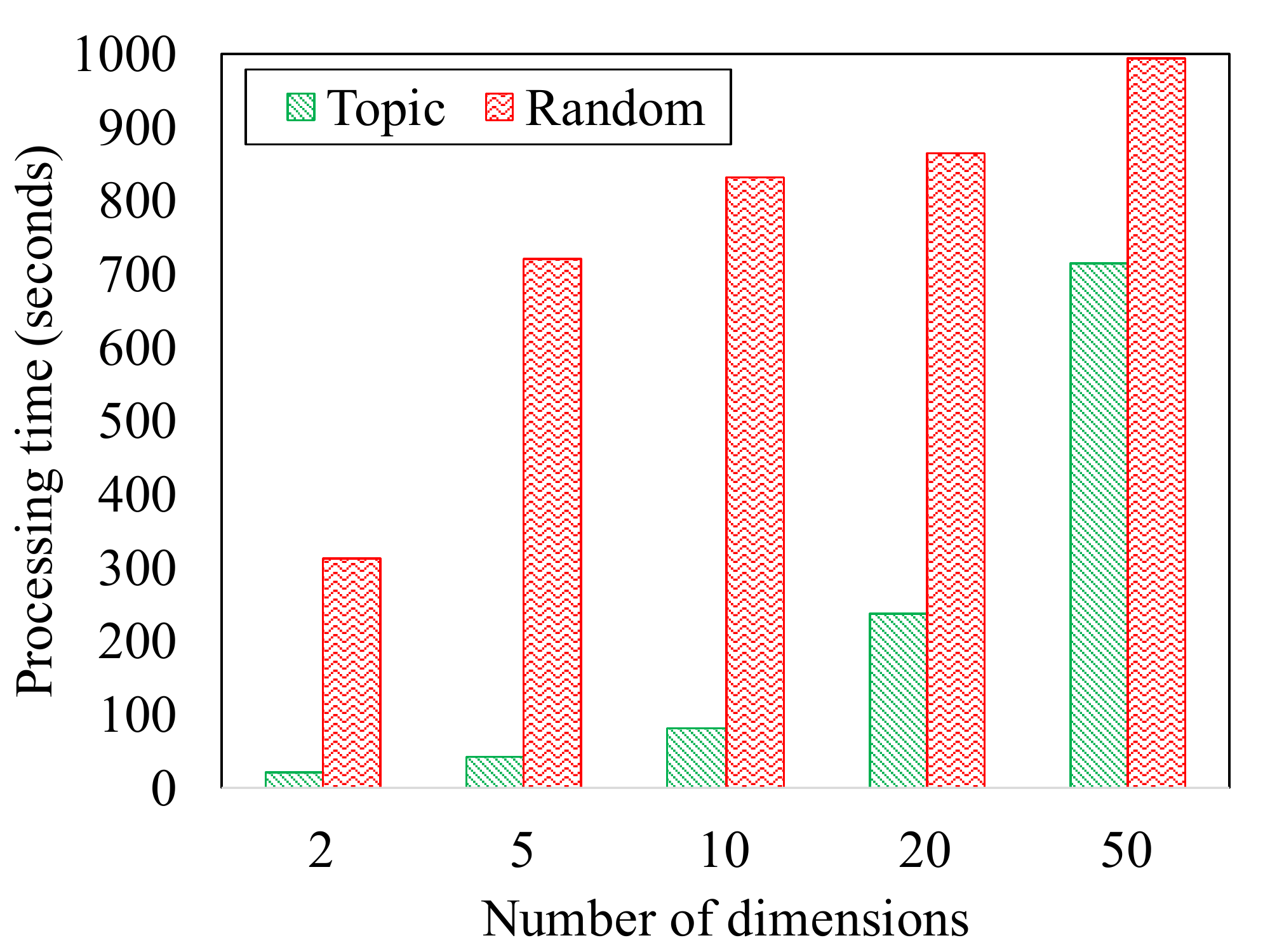}
        \caption{\scriptsize {($\mathtt{BPIC}$)}}
		\label{fig:bpic_time}
    \end{subfigure}
     \begin{subfigure}[t]{0.32\linewidth}
         \centering
         \includegraphics[width=1\textwidth]{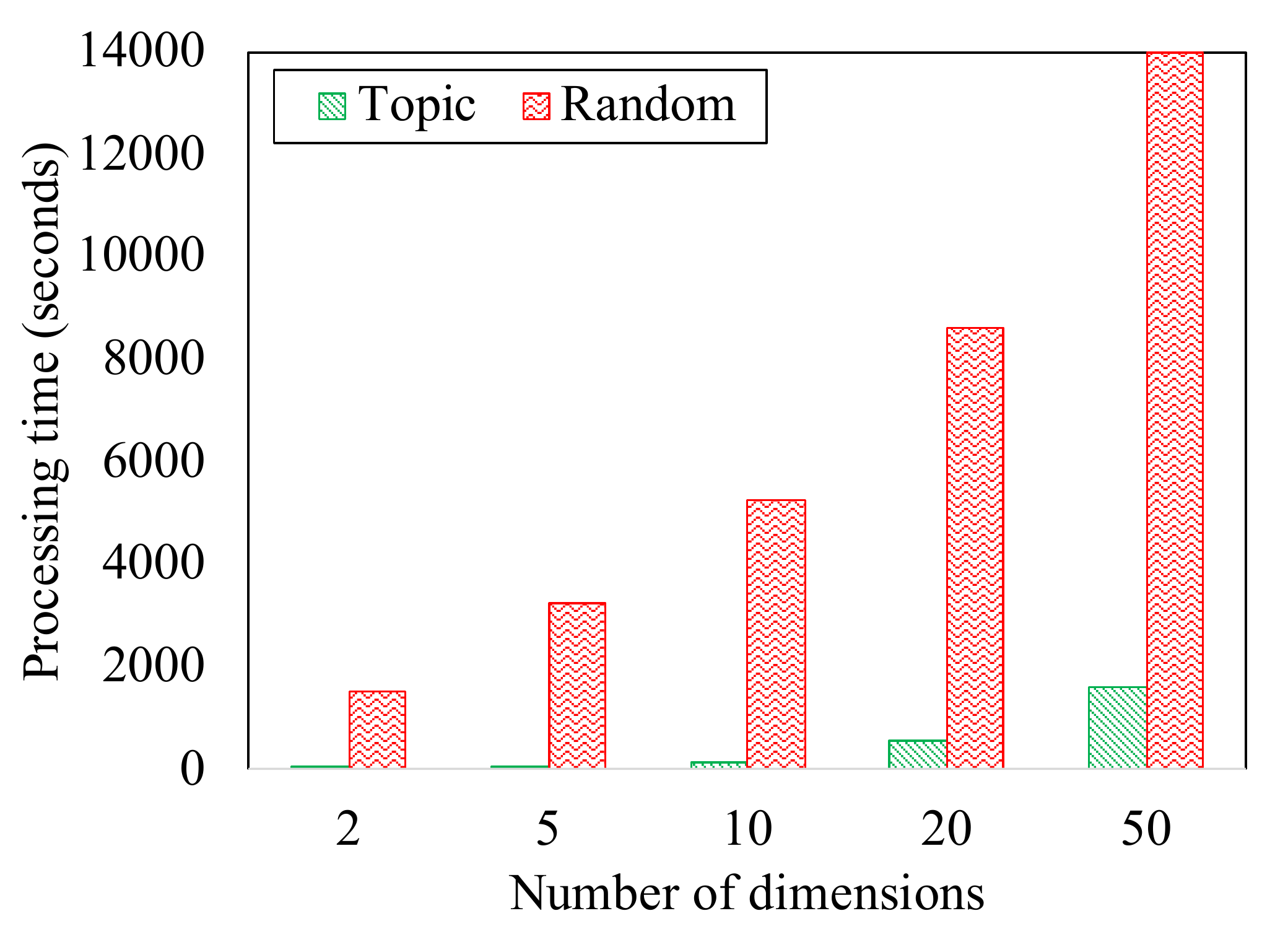}
         \caption{($\mathtt{Bank}$)}
	 	\label{fig:bank_time}
     \end{subfigure}
	\caption{Efficiency comparison of processing time between $\mathtt{Random}$ and $\mathtt{Topic}$ summarizations using the $\mathtt{Lithography}$, $\mathtt{BPIC}$, and $\mathtt{Bank}$ datasets.}
     \label{fig:efficiency_all}
\vspace{-0.2in}
\end{figure}

{\noindent \textbf{\textit{Evaluation metrics for the traces clustering task:}}} We evaluate the clustering results using process-specific metrics \cite{bose2009context}\cite{de2013active}: weighted average conformance fitness, and weighted average structure complexity. While the process model's conformance fitness quantifies the extent to which the discovered model can accurately reproduce the recorded traces, the structure complexity quantifies whether the clustering results produce process models that are simple and compact. Given a summarization scheme, we first transform all sequences to the summary space, and then perform traces clustering (using hierarchical clustering) with edit-distance as the similarity measure. Then, a process model is generated for each cluster using the Heuristic mining algorithm \cite{weijters2006process} and then converted to the Petri-Net model for conformance analysis. Given the Petri-net model, we use two publicly available plugins from the ProM framework \cite{van2005prom} for fitness and structural complexity analysis: The Conformance Checker Plugin is used to measure the fitness of the generated process models and the Petri-Net Complexity Analysis Plugin is used to analyze the structural complexity of the process models. After fitness and complexity scores are calculated for each cluster, the final scores are calculated as the average score over all clusters, weighted by the cluster size.

{\noindent \textbf{Effectiveness of summarization schemes on traces clustering:} %For the traces clustering task, we first compare the \textit{conformance fitness} of the clustering results in the summary space by different summarization schemes\footnote{\scriptsize We use $k=2$ for $\mathtt{Random}$, and $k=20$ for $\mathtt{Topic}$, as the two configurations share similar effectiveness in the similarity search task.} on $\mathtt{Lithography}$ dataset. Figure~\ref{fig:fitness} shows that, 
Figure~\ref{fig:fitness} highlights the \textit{conformance fitness} of the clustering results in the summary space by different summarization schemes\footnote{\scriptsize We use $k=2$ for $\mathtt{Random}$, and $k=20$ for $\mathtt{Topic}$, as the two configurations share similar effectiveness in the similarity search task.} on the $\mathtt{Lithography}$ dataset. 
Surprisingly, using summarization schemes not only helps improve the efficiency of the clustering task (as we showed earlier in the efficiency evaluation), but also helps produce clusters with process models of higher fitness, compared with the clustering results in the original space. The trend is similar when varying the number of clusters $N$. That is because measuring trace similarity on the summary space helps remove noise that often exists when measuring similarity using the original representation. 
Among summarization schemes, $\mathtt{Attribute}$ helps produce clustering results of higher conformance fitness (especially when using the $Tracked By$ attribute). That is because $\mathtt{Attribute}$ summarizations capture better the semantic relationship between traces (e.g., traces are similar if the corresponding sequences of $Sector$, $Tool$, or $Tracked By$ are similar).

 \begin{figure}[tp]
  %\begin{subfigure}{0.35\linewidth}
    \begin{minipage}{0.35\textwidth}
    \centering
   \includegraphics[width=1\linewidth]{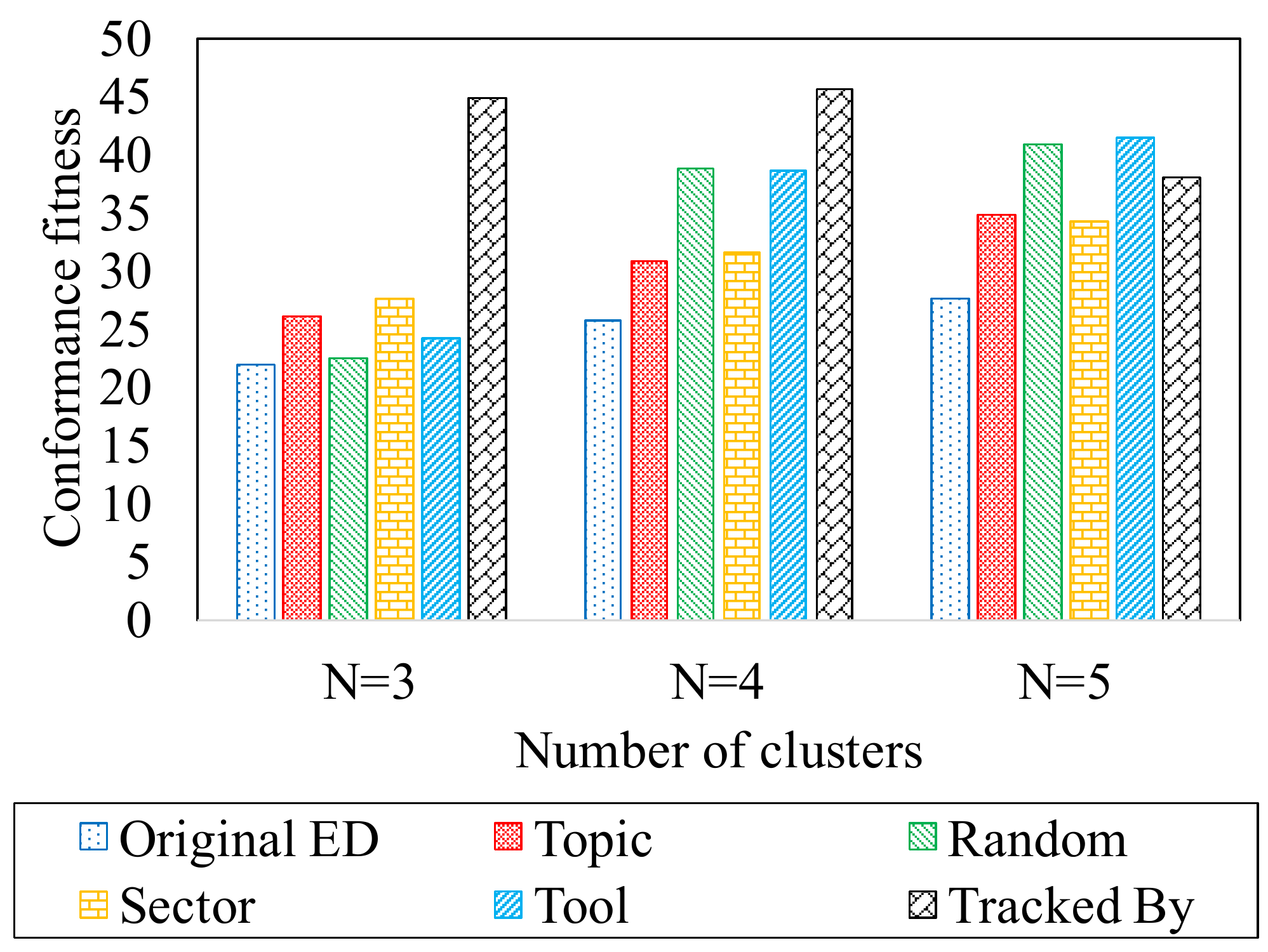}
 \caption{Conformance fitness comparison.}
 \label{fig:fitness}
%\end{subfigure}
\end{minipage}
~
%\begin{subfigure}[tp]{0.61\linewidth}
  \begin{minipage}{0.63\textwidth}
  \centering
   \includegraphics[width=1\linewidth]{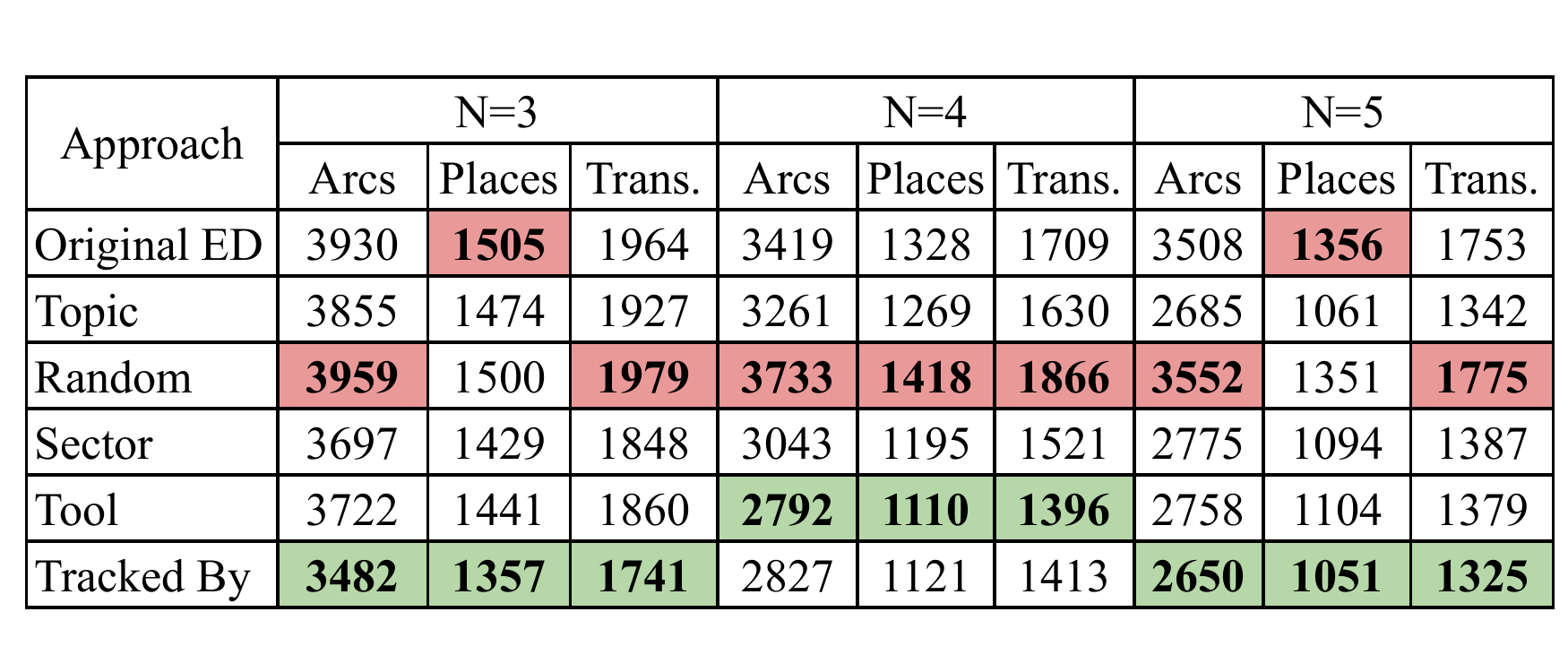}
 \caption{Traces clustering results' structural complexity comparison. (Green and red boxes denote best and worst results, respectively.)}
 \label{fig:complexity}
%\end{subfigure}
\end{minipage}
\vspace{-0.2in}
\end{figure}

In terms of the structural complexity (Figure~\ref{fig:complexity}), the $\mathtt{Attribute}$ summarizations outperform other summarization schemes and the results in the original space. This is again due to $\mathtt{Attribute}$'s ability to capture semantic relationships between traces, and thus, it helps produce clusters whose process models capture actual groups of traces that share similar semantic (and thus, have simple model structure). On the other hand, $\mathtt{Random}$ is the worst performer, due to the fact that random summarization could not capture the semantic relationship between traces.

In both conformance fitness and structural complexity tests, the $\mathtt{Topic}$ summarization produces results that approach that of $\mathtt{Attribute}$. Unlike $\mathtt{Attribute}$ summarization, which does not give users control over the resolution of the summaries, $\mathtt{Topic}$ summarization provides a qualitative advantage in offering a tunable parameter, $k$, to trade-off between the effectiveness and efficiency in the analysis task.

%It again shows that $\mathtt{Topic}$ is the choice of approach when it comes to flexible trade-off between the effectiveness and efficiency of the analysis task (i.e., similar to similarity search task, we can tune $k$ to achieve better efficiency for clustering task).

%\subsection{Efficiency of Summarization Schemes on Similarity Calculation}

\section{Conclusions and Future Works}
\label{sec:conclusions}
In this work, we introduce \textsc{Summarized}, a framework to perform efficient analysis on sequence-based multi-dimensional data using intuitive and user-controlled summarizations. We define a set of summarization schemes that offer flexible trade-off between quality and efficiency of analysis tasks and derive an error model for summary-based similarity under an edit-distance constraint. Evaluation results on real-world datasets show the effectiveness and efficiency of \textsc{Summarized}. For future work, we plan to apply our framework to other application domains, and design our framework to run on distributed infrastructure (e.g., Map-Reduce, Spark).

\vspace{-0.1 in}
\bibliographystyle{splncs04}
\bibliography{summarized}
\end{document}